%% file: main.tex
\documentclass[twocolumn]{svjour3}

\smartqed

\usepackage{tikz}
\usepackage{amsmath}
\usepackage{amssymb}
\usepackage{xifthen}
\usepackage{varwidth} 

\usepackage{hyperref}
\usepackage[T1]{fontenc}
\usepackage[utf8]{inputenc}
\usepackage{listings}
\usepackage{booktabs} 
\usepackage{colortbl}

\usepackage{acro}
\usepackage{xspace}

\usepackage{algorithm}
\usepackage[noend]{algorithmic}

\input{acronyms}
\input{macros}

 \usepackage[color=green]{todonotes}

\begin{document}

\title{
The Complexity of Data-Free Nfer
\thanks{Supported by Digital Research Centre Denmark (DIREC) and the Natural Sciences and Engineering Research Council of Canada (NSERC) [ref. RGPIN-2024-04280].}}

\author{Sean Kauffman \and
        Kim G.\ Larsen \and
		Martin Zimmermann}

\institute{Sean Kauffman \at
              Queen's University, Kingston, Canada\\
              \email{sean.k@queensu.ca}
           \and
           Kim G.\ Larsen \at
              Aalborg University, Aalborg, Denmark\\
              \email{kgl@cs.aau.dk}
		\and
	Martin Zimmermann \at
	              Aalborg University, Aalborg, Denmark\\
\email{mzi@cs.aau.dk}
}
\date{Received: date / Accepted: date}

\maketitle

\begin{abstract}
Nfer is a Runtime Verification language for the analysis of event traces that applies rules to create hierarchies of time intervals.
This work examines the complexity of the evaluation and satisfiability problems for the data-free fragment of nfer.
The evaluation problem asks whether a given interval is generated by applying rules to a known input, while the satisfiability problem asks if an input exists that will generate a given interval.

By excluding data from the language, we obtain polynomial-time algorithms for the evaluation problem and for satisfiability when only considering inclusive rules.
Furthermore, we show decidability for the satisfiability problem for cycle-free specifications with a NExpTime lower bound and undecidability for satisfiability of full data-free nfer.
\end{abstract}

\section{Introduction}

\Nfer is an interval logic for analyzing and comprehending event traces that has been used in a wide range of applications, from anomaly detection in autonomous vehicles~\cite{kauffman2020palisade} to spacecraft telemetry analysis~\cite{kauffman2017inferring}.
However, its high complexity demands that users restrict the features they incorporate into their applications to ensure tractability.
Despite this, no work exists that examines the runtime complexity of \nfer without data; an obvious restraint on the power of the language that more closely resembles propositional interval logics like \ac{hs}~\cite{halpern1991propositional} and Duration Calculus~\cite{goranko2004roadmap}.
These languages still tend to be undecidable in the general case, however, so it is unclear if this restriction on \nfer helps with tractability.
In this article, we show that evaluation of the data-free variant of \nfer is tractable and, futhermore, satisfiability for this variant is tractable with additional restrictions.

\Nfer was developed by scientists from NASA's \ac{jpl} in collaboration with other researchers to analyze event traces from remote systems like spacecraft~\cite{kauffman2016towards,kauffman2016nfer,kauffman2017inferring}.
In \nfer, specifications consist of rules that describe and label relationships between time periods called intervals.
Applying \nfer rules to an event trace yields a hierarchy of these intervals that is easier for humans and machines to comprehend than the raw events.

\Nfer originated from prior work at \ac{jpl} on trace analysis using rule-based programming systems~\cite{havelund2014comprehension}.
Existing systems such as the Scala DSL, Logfire~\cite{havelund2015logfire} 
had been used at NASA for some time and other, more similar, systems like PyContract~\cite{dams2022pycontract} are still in use.
However, these are more general tools for data analysis with tight integrations with general-purpose programming languages, and users have sometimes found them too verbose for many tasks~\cite{havelund2014experience}.
\Nfer is also a rule-based formalism like these tools, but is far more specialized for representing traces as hierarchies of intervals, leading to a more concise syntax for that task.

The conciseness of \nfer can be seen in the motivating example of ``warning analysis'' for the Curiosity rover (Mars Science Laboratory)~\cite{kauffman2017inferring}.
On the rover, events can be mis-labeled as \textsc{warning}s when a thread servicing the radio is starved by higher priority threads.
Previously, an ad-hoc Python script had been used to filter these anomalous \textsc{warning}s, but this proved to be error prone, slow, and hard to maintain.
\Nfer required only five rules to search for conditions where the benign \textsc{warning}s were known to occur and mark them as \textsc{ok}.

\Nfer resembles data-flow languages such as those found in \ac{cep} and \ac{srv} systems, which compute information about streaming data using predefined rules.
Like \nfer, \Acp{cep} process low-level event data into high-level abstractions~\cite{cugola2012processing}, but these tend to be more general-purpose, with more verbose rules and slower performance.
For example, Siddhi~\cite{suhothayan2011siddhi} is an open-source \ac{cep} used by Uber, eBay, and PayPal for use cases like fraud analysis and policy enforcement.
\Acp{cep} typically employ a SQL-like rule syntax applied to event streams to query arbitrary abstractions, which makes them powerful but generic tools, prone to verbose rules and slower response times.
In one experiment, the Python embedding of \nfer was found to be more than 30 times faster than Siddhi~\cite{kauffman2020palisade}.

\Ac{srv} systems support computing arbitrary streams of outputs from streams of input data and include examples such as TeSSLa~\cite{convent2018tessla}, Lola~\cite{dangelo2005lola}, and Striver~\cite{gorostiaga2018StriverStreamRuntime}.
\Ac{srv} languages tend to be good at computing streams of values, but are not as well suited to creating abstractions of events.
One study found it took 34 lines of TeSSLa to compute the result of four lines of \nfer~\cite{kauffman2021nfer}.
The data-flow language \acp{qre}~\cite{alur2016qre} is simpler and achieves linear-time evaluation complexity, but only supports the computation of a single quantitative result from a stream of events, making it much less expressive and unable to handle a use case such as the warning analysis described above.

The need for both expressiveness and performant execution makes the complexity problem for \nfer evaluation interesting.
Recent work analyzing the problem has shown that it is undecidable for the full language, but with various decidable fragments~\cite{kauffman2022complexity,kauffman2024complexity}.
Introducing less expressive language fragments for tractability purposes is a common technique for data-flow languages \cite{leucker2018TeSSLaRuntimeVerification,BAUMEISTER2026103495}.
These \nfer-language fragments mostly remain intractable, however, with \ptime complexity only possible by employing a meta-constraint on the size of the results that may not be practical in many cases.
One obvious restriction omitted by those works, however, was data-free \nfer, which has some precedent in the literature.

Like data-flow languages, \nfer typically operates on intervals with data, but here we define a data-free fragment of the language.
Data-free \nfer is expressive enough for many use cases, having been used, for example, to analyze the Sequential Sense-Process-Send (SSPS) dataset~\cite{kauffman2017inferring}.
Data-free \nfer is also the target for an algorithm to mine rules from real-time embedded systems~\cite{kauffman2017mining}.
What has not been known until now is how much of the expressive power of \nfer stems from its ability to analyze data and how much is inherent in its temporal interval reasoning capabilities.
%

Other works have examined the evaluation complexity of interval-based logics without data.
Halpern~\etal introduced an interval temporal logic and explored its decidability in~\cite{halpern1983itl}.
Chao\-chen~\etal found decidable and undecidable fragments of an extension of that work, Duration Calculus~\cite{CHAOCHEN1991269}, over both discrete and dense time~\cite{chaochen1993duration}.
Bolander~\etal later introduced \ac{hdc} that added the ability to name an interval and refer to it in a formula~\cite{BOLANDER2007113}.
They showed that \ac{hdc} can express Allen's relations and is decidable over discrete and dense time domains with non-elementary complexity.
Montanari~\etal examined the satisfiability problem for the subset of \ac{hs} over the natural numbers with only \emph{begins/begun by} and \emph{meets} operators and proved it to be \expspace-complete~\cite{montanari2010abb}.
The same authors later proved that adding the \emph{met by} operator increases the complexity of the language to be decidable only over finite total orders~\cite{montanari2010maximal}.
Aceto~\etal later examined the expressive power of all fragments of \ac{hs} over total orders~\cite{aceto2016classification}.

A major advantage of restricting \nfer to the data-free fragment is that the satisfiability problem also becomes interesting.
Here, the satisfiability problem asks if a set of \nfer rules can produce an interval with a particular identifier, restricting the inputs to events with a set of possible identifiers.
With data, it is trivial to show that satisfiability for \nfer is undecidable.
This follows from the results in~\cite{kauffman2024complexity} where \nfer is shown to have undecidable evaluation.
One can encode a Turing machine using \nfer rules and satisfiability asks if there is an initial tape such that the machine terminates.

Asking if an \nfer specification is satisfiable is useful in practice as it allows for a sanity-check on the validity of the rules.
If a user defines a set of rules that is meant to produce a certain interval if an unlikely problem occurs, it is helpful to know that the rules can, in fact, produce that interval.
If the rules are unsatisfiable, the user is likely to spend a long time debugging why the interval is not produced or, even worse, not notice the problem and deploy faulty rules to production.

Without data, it is much less obvious if \nfer satisfiability is undecidable.
In fact, we show that satisfiability for the full data-free \nfer language is still undecidable, but we achieve decidability by restricting to a cycle-free or inclusive-only fragment, the latter of which we demonstrate is decidable in \ptime.
That the satisfiability of inclusive, data-free \nfer is decidable in \ptime has exciting implications for practitioners, since these checks can be implemented in event-trace analysis tools~\cite{kauffman2023software,sadman2025visualizing}.
We also show that the evaluation problem for data-free \nfer is in \ptime without any artificial restrictions on the size of the result from meta-constraints.

This article is a revised and extended version of a paper published at RV~2024~\cite{DBLP:conf/rv/KauffmanLZ24}, which contains all proofs omitted in the conference version and improved results on cycle-free data-free \nfer satisfiability.

\subsection{Related Work on the Complexity of Interval Logics}
\label{sec:related}

\section{Data-Free \texorpdfstring{\nfer}{nfer}}
\label{sec:nfer}

We denote the set of nonnegative integers as $\natty$. 
The set of Booleans is given as $\boolty = \{\true, \false\}$.
We fix a finite alphabet~$\alphabet$ of event identifiers and a finite alphabet~$\Ident$ of interval identifiers such that $\alphabet \subseteq \Ident$.
A word is a sequence of identifiers $\utword = \utword_0 \utword_1 \cdots \utword_{|\utword|-1}$ where $\utword_i \in \alphabet$.
Given a word $\utword$, we define the non-empty subsequence $\utword_{[s,e]} = \utword_s \cdots \utword_e$, where $0 \leq s \leq e \le |\utword|-1$.

An event represents a named state change in an observed system.
An event is a pair $(\eta, t)$ where $\eta \in \alphabet$ is its identifier and $t \in \clockty$ is the timestamp when it occurred.
The set of all events is $\eventty = \alphabet \times \clockty$.
A trace is a sequence of events $\tau = (\eta_0, t_0)(\eta_1, t_1) \cdots (\eta_{n-1},t_{n-1})$ where $n = |\tau|$ and $t_i \leq t_j$ for all $i < j$.

Intervals represent a named period of state in an observed system.
An interval is a triple $(\eta, s, e)$ where $\eta \in \Ident$ is its identifier, and $s, e \in \clockty$ are the starting and ending timestamps where $s \leq e$.
We denote the set of all intervals by $\intervalty$. 
A set of intervals is called a \emph{pool} and the set of all pools is $\poolty = \powerset{\intervalty}$.
We say that an interval~$i = (\eta, s, e)$ is labeled by $\eta$ and define the accessor functions $\Idof{i} = \eta$, $\Startof{i} = s$, and $\Endof{i} = e$.
An interval of duration zero is an \emph{atomic} interval.

\subsection{Syntax}
The data-free \nfer syntax consists of \emph{rules}.
There are two forms of rules: inclusive and exclusive.
Inclusive rules test for the existence of two intervals matching a temporal constraint.
Exclusive rules test for the existence of one interval and the absence of another interval matching a temporal constraint.
When such a pair is found, a new interval is produced with an identifier specified by the rule and timestamps taken from the matched intervals.
We define the syntax of these rules as follows:
\begin{itemize}
    \item Inclusive rules have the form~$\dfnferrule{\eta}{\eta_1}{\oplus}{\eta_2}$ and
    \item exclusive rules have the form~$\dfnferrule{\eta}{\eta_1}{\unlesskw \ominus}{\eta_2}$
\end{itemize}

\noindent
where $\eta, \eta_1, \eta_2 \in \Ident$ are identifiers,  
\begin{multline*}
  \oplus \in \{\beforekw{}, \meetkw{}, \duringkw{}, \coincidekw{}, \startkw{},  \finishkw{}, \\ \overlapkw{}, \slicekw{}\}  
\end{multline*}
is a \emph{clock predicate} on three intervals (one for each of $\eta,\eta_1,\text{ and } \eta_2$), and  \[\ominus \in \{\afterkw{},\followkw{},\containkw{}\}\] is a clock predicate on two intervals (one for each of $\eta_1$ and $\eta_2$).
%
For a rule~$\dfnferrule{\eta}{\eta_1}{\oplus}{\eta_2}$ or $\dfnferrule{\eta}{\eta_1}{\unlesskw \ominus}{\eta_2}$ we say that $\eta$ appears on the left-hand and the $\eta_i$ appear on the right-hand side.

A specification  is a finite list of rules.

\subsection{Semantics}
The semantics of the \nfer language is defined in three steps: the semantics~$R$ of individual rules on pools, the semantics $S$ of a specification (a list of rules) on pools, and the semantics $T$ of a specification on traces of events.

We first define the semantics of inclusive rules with the interpretation function $R$.
Let $\rulety$ be the set of all rules.
Semantic functions are defined using the brackets \isem{[:_:]} around syntax being given semantics.

\begin{lstlisting}
RuleSem[:_:] : Delta -> Pool -> Pool
RuleSem[:eta <- eta1 oplus eta2:]pi = { i0 isin Interval : Exists i1,i2 isin  pi :- 
  Idof(i0) = eta &&Idof(i1) = eta1 &&Idof(i2) = eta2 &&oplus(i0,i1,i2) }
\end{lstlisting}
In the definition, an interval $i$ is a member of the produced pool when two existing intervals in $\pool$ match the identifiers $\eta_1$ and $\eta_2$ and the temporal constraint $\oplus$.
The identifier of $i$ is given in the rule and $\oplus$ defines its start and end timestamps.

The clock predicates referenced by $\oplus$ are shown in Table~\ref{tab:inc-constraints}.
These relate two intervals using the familiar \ac{atl} temporal operators~\cite{allen1983maintaining} and also specify the start and end timestamps of the produced intervals.
For the example $\beforekw{}(i,i_1,i_2)$, $i_1$ and $i_2$ are matched when $i_1$ ends \beforekw{} $i_2$ begins.
The generated interval $i$ has start and end timestamps inherited from the intervals $i_1$ and $i_2$, i.e., no new timestamps are generated by applying $\beforekw{}(i,i_1,i_2)$. This is true for all other rules as well. 

\begin{table}
    \caption{Formal definition of \nfer clock predicates for inclusive rules}
    \label{tab:inc-constraints}
    \centering
    \renewcommand{\arraystretch}{2}
    \begin{tabular}{c c}
  \toprule
  $\oplus$ & Constraints on $i, i_1,$ and $i_2$ \\
  \midrule
  
  \beforekw{} & $\Endof{i_1} < \Startof{i_2} \wedge$ \\
   & $\Startof{i} = \Startof{i_1} \wedge \Endof{i} = \Endof{i_2}$ \\

\rowcolor{gray!30}  \meetkw{} & $\Endof{i_1} = \Startof{i_2} \wedge$ \\
\rowcolor{gray!30} & $\Startof{i} = \Startof{i_1} \wedge \Endof{i} = \Endof{i_2}$ \\

    \duringkw{} & $\Startof{i_2} = \Startof{i} \leq \Startof{i_1} \wedge$ \\
     & $\Endof{i_1} \leq \Endof{i_2} = \Endof{i}$ \\

\rowcolor{gray!30}    \coincidekw{} & $\Startof{i_1} = \Startof{i_2} = \Startof{i} \wedge$ \\
\rowcolor{gray!30}  & $\Endof{i_1} = \Endof{i_2} = \Endof{i}$ \\

    \startkw{} & $\Startof{i_1} = \Startof{i_2} = \Startof{i}\: \wedge $\\
     & $\Endof{i} = \max(\Endof{i_1},\Endof{i_2})$ \\

\rowcolor{gray!30}  \finishkw{} & $\Endof{i} = \Endof{i_1} = \Endof{i_2}\: \wedge $ \\
\rowcolor{gray!30} & $\Startof{i} = \min(\Startof{i_1},\Startof{i_2})$ \\

  \overlapkw{} &  $\Startof{i_1} < \Endof{i_2} \wedge \Startof{i_2} < \Endof{i_1}\ \wedge$  \\ 
  & $\Startof{i} = \min(\Startof{i_1},\Startof{i_2}) \wedge$ \\ 
  & $\Endof{i} = \max(\Endof{i_1},\Endof{i_2})$ \\

 \rowcolor{gray!30} \slicekw{} &    $\Startof{i_1} < \Endof{i_2} \wedge \Startof{i_2} < \Endof{i_1}\ \wedge$ \\ 
 \rowcolor{gray!30} & $\Startof{i} = \max(\Startof{i_1},\Startof{i_2}) \wedge$ \\ 
 \rowcolor{gray!30} & $\Endof{i} = \min(\Endof{i_1},\Endof{i_2})$ \\

    \end{tabular}
\end{table}

In Figure~\ref{fig:inclusiveops}, the semantics of the inclusive rules is graphically represented. There, an interval is represented by a line with its label written above, time moves from left to right, and the $B$-labeled interval is drawn only once and shared by each example.
For the example $\dfnferrule{A}{B}{\beforekw{}}{C}$, the $B$-labeled and $C$-labeled intervals are matched because the $B$-labeled interval ends \beforekw{} the $C$-labeled interval begins.

\begin{figure}
    \centering
    \includegraphics[width=.45\textwidth]{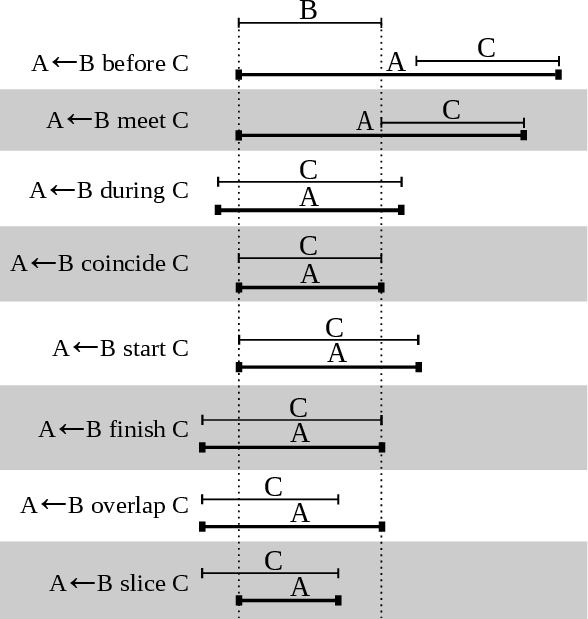}%
    \caption{\Nfer clock predicates ($\oplus$) for inclusive rules}
    \label{fig:inclusiveops}
\end{figure}

\noindent
We now define the semantics of exclusive rules with the function $R$.
\begin{lstlisting}
RuleSem[:eta <- eta1 unless ominus eta2:]pi = { i0 isin Interval : Exists i1 isin pi :- 
  Idof(i0) = eta &&Idof(i1) = eta1 &&
  Startof(i0) = Startof(i1) && Endof(i0) = Endof(i1) &&
  not (Exists i2 isin pi :- i2 != i1 &&Idof(i2) = eta2 &&ominus(i1,i2) ) }
\end{lstlisting}

Like with inclusive rules, exclusive rules match intervals in the input pool~$\pi$ to produce a pool of new intervals.
The difference is that exclusive rules produce new intervals where one existing interval in $\pi$ matches the identifier~$\eta_1$ and no intervals exist in $\pi$ that match the identifier~$\eta_2$ such that the clock predicate~$\ominus$ holds for the $\eta_1$-labeled and the $\eta_2$-labeled interval. 

The three possibilities referenced by $\ominus$ are shown in Table~\ref{tab:exc-constraints}.
These clock predicates relate two intervals using familiar \ac{atl} temporal operators while the timestamps of the produced interval are copied from the included interval rather than being defined by the clock predicate.
For the example $\afterkw{}(i_1,i_2)$, $i_1$ and $i_2$ would be matched (if $i_2$ existed) if $i_1$ begins after $i_2$ ends, and this match would result in \emph{no} interval being produced.
If such an interval~$i_2$ is absent, an interval is produced with timestamps matching $i_1$.

\begin{table}[]
  \caption{Formal definition of \nfer clock predicates for exclusive rules}
  \label{tab:exc-constraints}
  \centering
  \renewcommand{\arraystretch}{2}
  \begin{tabular}{c cc}
  \toprule
  $\ominus$ & \phantom{} & Constraints on $i_1$ and $i_2$ \\
  \midrule

  \afterkw{} & \phantom{} & $\Startof{i_1} > \Endof{i_2} $ \\
 \rowcolor{gray!30} \followkw{} & \phantom{} & $\Startof{i_1} = \Endof{i_2} $\\
  \containkw{} & \phantom{} & $\Startof{i_2} \geq \Startof{i_1}  \wedge \Endof{i_2} \leq \Endof{i_1}$ \\
  \end{tabular}
\end{table}

In Figure~\ref{fig:exclusiveops}, the semantics of the exclusive rules is graphically represented. There, an interval is represented by a line with its label written above, time moves from left to right, and the $B$-labeled interval is drawn only once and shared by each example.
For the example $\dfnferrule{A}{B}{\unlesskw{}\ \afterkw{}}{C}$, the $B$-labeled and $C$-labeled intervals would be matched (if the $C$-labeled interval existed), which would result in the $A$-labeled interval \emph{not} being produced.

\begin{figure}[ht]
    \centering
    \includegraphics[width=.45\textwidth]{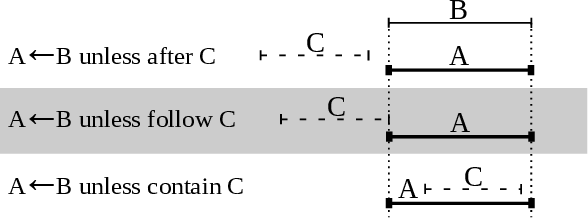}%
    \caption{\Nfer clock predicates ($\ominus$) for exclusive rules}
    \label{fig:exclusiveops}
\end{figure}



The interpretation function $S$ defines the semantics of a specification.
Given a specification~$\delta_1 \cdots \delta_n \in \rulety^*$ and a pool $\pool \in \poolty$, $\seminterp{\_}$ recursively applies $\interp{\_}$ to the rules in order, passing each instance the union of $\pool$ with the intervals returned by already completed calls.
\begin{align*}
&S\;[\![\_]\!]\ :\ \rulety^* \rightarrow \poolty \rightarrow \poolty \\
&S\;[\![\ \delta_1 \cdots \delta_n\ ]\!]\ \pi\ =\\ 
&\hspace{2em}\begin{cases}
  S\;[\![\ \delta_2 \cdots \delta_n\ ]\!]\ (\,\pi\; \cup\; R\;[\![\ \delta_1\ ]\!]\ \pi\ ) & \textbf{if } n > 0\\
  \pi & \textbf{otherwise}
\end{cases}
\end{align*}

An \nfer specification~$D \in \rulety^*$ forms a directed graph $G(D)$ with vertices for the rules in $D$ connected by edges representing identifier dependencies.
An edge exists in $G(D)$ from $\delta$ to $\delta'$ iff there is an identifier~$\eta$ that appears on the left-hand side of $\delta$ and the right-hand side of $\delta'$.
We say that $D$ contains a cycle if $G(D)$ contains one; otherwise $D$ is cycle-free.

The rules in a cycle in an \nfer specification must be iteratively evaluated until a fixed point is reached.
As intervals may never be destroyed by rule evaluation, inclusive rules may be repeatedly evaluated, safely.
However, exclusive rules may not be evaluated until the intervals on which they depend are known to be present or absent.

\begin{example}
Suppose a specification with the two rules \ $\delta_1 = \dfmeetrule{c}{a}{b}$ and $\delta_2 = \dfmeetrule{a}{c}{b}$.
Given $\pi = \{ (a,0,1),(b,1,2),(b,2,3),$ $(b,3,4),(d,4,5) \}$,  we have\linebreak $\interp{\delta_1}{\pi} = \{ (c,0,2) \}$ and $\interp{\delta_2}{(\pi \cup \interp{\delta_1}{\pi})} = \{ (a,0,3) \} $.
The rules must be applied a second time to reach a fixed point that includes the interval $(c,0,4)$.

Consider the consequences if the specification also contained the exclusive rule $\delta_3 = $ $\dffollowrule{b}{d}{c}$.
After the first evaluation, $(c,0,4)$ is not yet produced, so evaluating $\delta_3$ would generate $(b,4,5)$, an incorrect result.
\end{example}

As such, \emph{exclusive rules may not appear in cycles} but may appear in a specification that contains cycles among inclusive rules.
To find the cycles in a specification, we compute the strongly-connected components of the directed graph $G(D)$ formed by the rules in $D$.
Each strongly connected component represents either a cycle or an individual rule outside of a cycle.
We then sort the components in topological order and iterate over each component until a fixed point is reached.

The interpretation function $\traceinterp{\_}{}$ defines the semantics of a specification applied to a trace of events.
To ensure consistency with prior work and to simplify our presentation, we overload $\traceinterp{\_}{}$ to operate on an event trace~${\tau \in \eventty^*}$ by first converting $\tau$ to the pool $\{ \initf{e} : e \text{ is an element of } \tau \}$ where $\initf{\eta,t} = (\eta,t,t)$.

\begin{lstlisting}
TraceSem[:_:] : DeltaList -> Pool -> Pool
TraceSem[:delta1 ... deltaN:]pi = piSIZEl :- pi1 = pi && 
    Comps = SCC(delta1 ... deltaN) && (D1 ... Dl)  = topsort(Comps) :- 
    piO = UNIONAllJ piJ:-piK = SpecSem[:Di:](piJ)
\end{lstlisting}
where $\text{SCC}(\delta_1 \cdots \delta_n)$ is the set $\mathcal{D}$ of strongly connected components of the graph $G(\delta_1 \cdots \delta_n)$  and $\textit{topsort}(\mathcal{D})$ is a topological sort of these components.


%

\section{Satisfiability of Data-Free \texorpdfstring{\nfer}{nfer}}
\label{sec:satisfiability}

We are interested in the existential \nfer satisfiability problem: Given a specification~$D$, a set of identifiers $\inputids$, and a target identifier~$\target$, is there an input trace of events $\tau \in \eventty^+$ with identifiers in $\inputids$ 
such that an $\target$-labeled interval is in $\traceinterp{D}{\tau}$?
The \nfer satisfiability problem is interesting in part because of the restriction of input identifiers to $\inputids \subseteq \Ident$.
If $\target \in \inputids$, then any specification is trivially satisfiable.
When $\target \not\in \inputids$, however, then an $\target$-labeled interval must be derived.
This problem is non-trivial and, as we shall see, undecidable in general.

\begin{example}
\label{example_running}
To see how data-free \nfer specifications may be satisfiable or not, consider the following two example specifications for the target identifier $\target$ and input identifiers $\inputids = \{a,b\}$:

\newcommand{\DSat}{D_{\text{sat}}}
\newcommand{\DUnsat}{D_{\text{unsat}}}
\vspace{1em}
\hbox{\hspace{1em}$\DSat =$
\begin{varwidth}{\textwidth}
\[\begin{cases}
 & \dfnferrule{A}{a}{\beforekw{}}{b} \\
 & \dfnferrule{B}{A}{\meetkw{}}{b} \\
 & \dfnferrule{\target}{A}{\overlapkw{}}{B} \\
\end{cases}\]
\end{varwidth}
}
\hbox{\hspace{1em}$\DUnsat =$
\begin{varwidth}{\textwidth}
\[\begin{cases}
 & \dfnferrule{A}{b}{\beforekw{}}{X} \\
 & \dfnferrule{B}{a}{\meetkw{}}{b} \\
 & \dfnferrule{\target}{a}{\overlapkw{}}{B} \\
\end{cases}\]
\end{varwidth}
}

\vspace{1em}
A satisfying event trace for $\DSat$ is $\tau_1 = (a,1)(b,2)$, since 
$\traceinterp{\DSat}{\tau_1} = \{ (a,1,1),(b,2,2), (A,1,2),(B,1,2),$ $(\target,1,2)\}.$
For $\DUnsat$, no $\target$-labeled interval can be produced because $\overlapkw{}$ requires one of the two matched intervals to have positive duration: for an interval $i$, $\Endof{i} - \Startof{i} > 0$.
Since $a$-labeled intervals must be initial, they are atomic (zero duration).
That leaves $B$-labeled intervals produced by another rule.
The rule that produces $B$-labeled intervals, however, only matches initial intervals with the same timestamps.
As such, any $B$-labeled interval will also have zero duration, and the $\overlapkw{}$ rule will never be matched.
Finally, the $\beforekw{}$ rule can never be applied, as no $X$-labeled interval can be generated.
\end{example}
\subsection{Data-Free \texorpdfstring{\nfer}{nfer} Satisfiability is Undecidable}
\label{sec:undecidability}

In this section, we show the undecidability of the data-free \nfer satisfiability problem by a reduction from the emptiness problem for the intersection of two \acp{cfg}.
The undecidability result relies on the recursive nature of \nfer, i.e., an $\eta$-labeled interval can be produced from another $\eta$-labeled interval, and on its negation capabilities, i.e., via exclusive rules.

\newcommand{\etaacc}[1]{\eta_{#1}}
\begin{theorem}
\label{thm:undecidable}
The data-free \nfer satisfiability problem is undecidable.
\end{theorem}

We now show how to simulate \iac{cfg}~$G$ with data-free \nfer rules with a designated identifier~$\etaacc{G}$ so that a word~$w$ is accepted by the \ac{cfg} iff applying the rules to events that correspond to $w$ generates an interval over the same period with identifier~$\etaacc{G}$.
Then the intersection of two \acp{cfg} $G_1$ and $G_2$ is nonempty if and only if applying the corresponding rules generates, for some sequence of events corresponding to a word, two intervals with the same starting and ending timestamps, one with identifier~$\etaacc{G_1}$ and one with $\etaacc{G_2}$.
The existence of two such intervals can again be captured by a data-free \nfer rule producing an interval with a target identifier.

Formally, \iac{cfg} is a four-tuple $(\nonterms, \alphabet, P, S)$, where $\nonterms$ is a finite set of non-terminals (or variables), $\alphabet$ is the finite set of terminals that are disjoint from $\nonterms$, $P$ is a finite set of productions of the form $A \rightarrow w$ where $A \in \nonterms$ and $w \in (\nonterms \cup \alphabet)^*$, and $S$ is the initial non-terminal.
We assume, without loss of generality, that a \ac{cfg} is in Chomsky-normal form~\cite{chomsky1959normal}.\footnote{Note that we, w.l.o.g., disregard the empty word.}
This means that all productions are in one of two forms: $A \rightarrow BC$ or $A \rightarrow a$ where $A,B,C \in \nonterms$, $a \in \alphabet$, and $S \notin \{B,C\}$.

Given a grammar $(\nonterms, \alphabet, P, S)$, where $A \in \nonterms$, $w,x,y \in (\nonterms \cup \alphabet)^*$, and $(A \rightarrow x) \in P$, then we say that $wAy$ yields $wxy$, written $wAy \Rightarrow wxy$.
We write $w \xRightarrow{*} y$ if $w = y$ or there exists a sequence of strings $x_1,x_2,\ldots,x_n$ for $n \geq 0$ such that $w \Rightarrow x_1 \Rightarrow x_2 \Rightarrow \cdots \Rightarrow x_n \Rightarrow y$.
The language of the grammar is $\{ w \in \alphabet^* \mid S \xRightarrow{*} w \}$.
If a word is in the language of the grammar we say it has a derivation in the grammar.
Deciding if the intersection of the languages of two \acp{cfg} is empty is undecidable~\cite{barhillel1961cfgs}.

\newcommand{\ex}{\textit{x}}
We use grammar~$G_{\ex} = (\nonterms_{\ex}, \alphabet_{\ex}, P_{\ex}, S_{\ex})$ accepting $\{a(a^n)(b^n) : n > 0\}$ as a running example, where
%
 $\nonterms_{\ex} = \{A,B,M,M',S_{\ex}\}$, $\alphabet_{\ex} = \{a,b\}$, and $P_{\ex} = \{ 
A \rightarrow  a,
B \rightarrow  b,
S_{\ex} \rightarrow  AM,
M \rightarrow AB,
M \rightarrow  AM',
M' \rightarrow  MB
\}$.
The string $aab$ has the derivation 
$S_{\ex} \Rightarrow AM \Rightarrow AAB \Rightarrow AAb \Rightarrow Aab \Rightarrow aab$.

We present six types of data-free \nfer rules to simulate the intersection of two \acp{cfg} $G = (\nonterms, \alphabet, P, S)$ and $G' = (\nonterms', \alphabet, P', S')$, where $\nonterms \cap \nonterms' = \varnothing$ and $P \cap P' = \varnothing$.
The first four steps are necessary to account for events with coincident timestamps and because simulating \iac{cfg} requires the sequential composition of non-terminals, while \nfer rules cannot perform sequential composition directly on atomic intervals.
Then the final two steps map the productions of \iac{cfg} and their intersection directly to data-free \nfer rules.
The six types of rules are:
\begin{enumerate}
    \item Rules to label non-unique timestamps in an event trace so that they can be filtered out.
    We do so, because event traces in \nfer are allowed to have events with the same timestamps while there is only one letter at each position of a word. 
    So, to simplify our translation between event traces and words, we just filter out events with non-unique timestamps. 
    \item Rules that then perform the actual event filtering to only include events with unique timestamps.
    \item Rules that label every interval in a trace by its starting event, i.e. where some event occurs at the start and some other event occurs at the end, we label the interval by the starting event.
    \item Rules that select the minimal starting-event-labeled intervals, i.e. the intervals where no other interval is subsumed by that interval.  The result of this step is a set of contiguous intervals labeled by their starting event. These minimal intervals are totally ordered and in one-to-one correspondence with the original events with unique timestamps. Thus, we have transformed the event trace into an \emph{equivalent} pool of intervals.
    \item Rules that simulate the productions of the \acp{cfg} on the pool of minimal starting-event-labeled intervals. The generated intervals  encode a derivation tree. The word corresponding to the event trace is accepted by \iac{cfg} if an interval is generated that is labeled by that grammar's initial non-terminal.
    \item A rule that labels an interval by a given target label if the simulation of the two \acp{cfg} labeled the same interval by their initial non-terminals.  The interval is generated if the word corresponding to the same event trace is accepted by both \acp{cfg}.
\end{enumerate}

We begin by relating event traces to words.
Event traces form a total preorder as some timestamps may be equal, while the symbols in a word are totally ordered by their index (no two symbols have the same index). 
To convert a word to an event trace, we add a timestamp equal to the index of the event.
Given a word $\utword \in \alphabet^*$, we define $\wordtotrace(\utword) = (\utword_0,0)(\utword_1,1)\cdots(\utword_{n-1},n-1) \text{ where } n = |\utword|$.

For example, $\wordtotrace(aab) = (a,0)(a,1)(b,2)$.

When converting an event trace to a word, however, we must only consider events with unique timestamps.
The following example trace illustrates the reasoning at this step: consider $\tau_{\ex} = (a,0)(a,1)(a,2)(b,2)(b,3)(b,4)$ where both an $a$-labeled and a $b$-labeled event occur at timestamp $2$.
To convert $\tau_{\ex}$ to a word, we want to order its identifiers using only their timestamps, and, consequently, the two events with timestamp $2$ cannot be ordered.
As such, we ensure that only events with unique timestamps affect the generation of intervals involved in simulating \iac{cfg}.
Before we show how to generate those intervals we define formally what we mean by unique timestamps.

Given a trace~$\tau = (\eta_0, t_0) (\eta_1,t_1)\cdots (\eta_{n-1},t_{n-1})$, $t_i$ is non-unique in $\tau$ if there is a $j \neq i$ with $\eta_i \neq \eta_j$ and $t_j = t_i$, otherwise $t_i$ is unique.\footnote{Note that the timestamp~$7$ in $(a, 3)(b,7)(b,7)(c,13)$ is unique, as the two events at timestamp~$7$ have the same label. However, this duplication will be removed when turning the trace into a pool, which is a set of events and therefore has no duplicates.}
Let $\unique(\tau) = \{ t_{i_0},t_{i_1},\ldots,t_{i_{k-1}} \}$ be the set of unique timestamps in $\tau$ such that $t_{i_j} < t_{i_{j'}}$ for all $j < j'$, i.e., we enumerate the unique timestamps of $\tau$ in increasing order.
Then, we define $\tracetoword(\tau) = \eta_{i_{0}} \eta_{i_{1}}\cdots \eta_{i_{k-1}}$.

We now define data-free \nfer rules that capture the definition of $\tracetoword$.
The rules first generate atomic $\spoil$ intervals where events with distinct identifiers share timestamps and then filter the events to only those that do not share timestamps with those $\spoil$ intervals.

Given the alphabet $\alphabet$, we define rules to generate $\spoil$ intervals for non-unique timestamps.
We let $\spoil$ be a new identifier ($\spoil \notin \alphabet$).
\begin{equation*}
D_1 = \{ \dfcoinciderule{\spoil}{a}{b} \mid (a, b) \in \alphabet \times \alphabet \wedge a\neq b \}
\end{equation*}
For example, applying $D_1$ to the example trace $\tau_{\ex}$ defined above results in the following intervals: 
\begin{align*}\traceinterp{D_1}{(\tau_{\ex})} = \{&(a,0,0),(a,1,1),(a,2,2),(b,2,2),\\&(\spoil,2,2),(b,3,3),(b,4,4)\}.
\end{align*}
Figure~\ref{fig:d1234} shows this example on a timeline.
In the top of the figure, the solid line shows time progressing from left to right, with the identifiers appearing in the trace given below their associated timestamps.
The new $\spoil$-labeled interval is shown below the timeline, having been generated by the rules in $D_1$ shown on the right.
The remainder of the figure relates to steps 2, 3, and~4.

\begin{figure*}
    \centering
    \includegraphics[width=.7\textwidth]{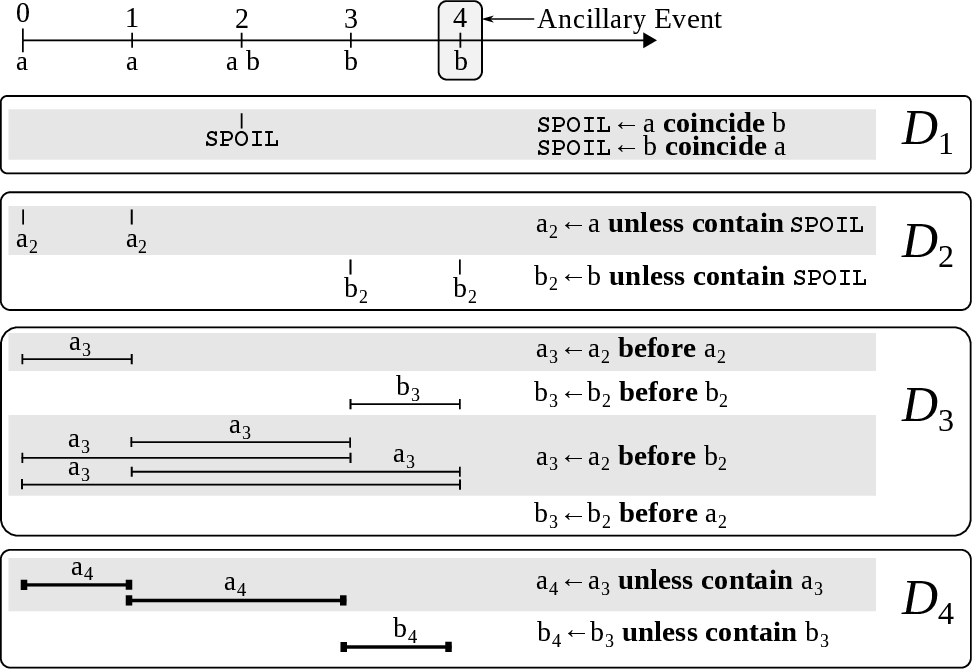}
    \caption{Example of applying steps 1-4 from the proof of Theorem~\ref{thm:undecidable}}
    \label{fig:d1234}
\end{figure*}

\begin{lemma}
\label{lemma_spoiling}
\hfill

\begin{enumerate}
    \item Given a trace $\tau$ over $\alphabet$, $\traceinterp{D_1}{(\tau)}$ characterizes the unique timestamps in $\tau$ in the following sense:
$\{ t \mid (\spoil,t,t) \in \traceinterp{D_1}{(\tau)} \} = T \setminus \unique(\tau)$, where $T$ denotes the set of (unique and non-unique) timestamps of $\tau$.

    \item Let $w \in \Sigma^*$ and let $\tau_w = \wordtotrace(w)$. Then, $\traceinterp{D_1}{(\tau_w)}$ does not contain any $\spoil$-labeled interval.
\end{enumerate}
\end{lemma}

\begin{proof}
The second item follows trivially from the definition of $\wordtotrace(w)$, as all timestamps are unique in it. 
Hence, we focus on the first item.

Let $(\spoil,t,t) \in \traceinterp{D_1}{(\tau)}$. 
As $\spoil \notin \alphabet$, the interval must have been created by applying a rule of the form~$\dfcoinciderule{\spoil}{a}{b}$ for some $a \neq b$ in $\Sigma$. 
Furthermore, as the rules in $D_1$ only create intervals with the identifier~$\spoil$, the intervals to which the rule has been applied to create~$(\spoil,t,t)$ must come from the initial pool~$\set{\initf{e} \mid e \text{ is an event in } \tau}$ on which $D_1$ is evaluated. 
Hence, $\tau$ must contain events~$(a,t)$ and $(b,t)$, i.e., $t$ is not unique in $\tau$.

Conversely, let $t$ be a non-unique timestamp in $\tau$, i.e., there are events~$(a,t)$ and $(b,t)$ in $\tau$ for some identifiers~$a \neq b$. 
These intervals are also in the initial pool~$\set{\initf{e} \mid e \text{ is an event in } \tau}$ on which $D_1$ is evaluated. 
Hence, applying the rule~$\dfcoinciderule{\spoil}{a}{b}$ to the intervals~$(a,t,t)$ and $(b,t,t)$ implies $(\spoil,t,t) \in \traceinterp{D_1}{(\tau)}$.
\qed\end{proof}

Now, we can define rules that filter the events to only those with unique timestamps by excluding any that coincide with $\spoil$-labeled intervals.
These rules ensure that the \nfer simulation of \iac{cfg} uses exactly the same events that are used in $\tracetoword$: those with unique timestamps.
Note that the intervals generated by the rules in $D_i$ for steps $i \in \{2,3,4\}$ are labeled by identifiers annotated by the step number, i.e., they are of the form~$a_i$ for $a \in \alphabet$ and $i \in \{2,3,4\}$, where we assume w.l.o.g., $a_i \notin \nonterms \cup \nonterms' \cup \alphabet$ for all such identifiers.
\begin{equation*}
D_2 = \{ \dfcontainrule{a_2}{a}{\spoil} \mid a \in \alphabet \}
\end{equation*}

Figure~\ref{fig:d1234} shows the result of applying $D_2$ to the result of $\traceinterp{D_1}{(\tau_{\ex})}$.
The intervals $(a_2,0,0)$ and $(a_2,1,1)$ annotate $a$-labeled intervals that do not coincide with a $\spoil$-labeled interval, while $(b_2,3,3)$ and $(b_2,4,4)$ annotate the $b$-labeled intervals that do not coincide with a $\spoil$-labeled interval.
No such annotated intervals are produced at timestamp 2, where the rules in $D_1$ generated a $\spoil$-labeled interval.

The intervals that result from the rules in $D_2$ are still atomic, i.e. they are effectively events and have a duration of zero.
The next step is to use those atomic intervals to generate every interval in the trace with a positive duration (restricted to those with unique starting and ending timestamps in the original trace~$\tau$).
We label every such interval with a label derived from its start.
\begin{equation*}
D_3 = \{ \dfbeforerule{a_3}{a_2}{b_2} \mid (a,b) \in \alphabet \times \alphabet \} 
\end{equation*}

As shown in Figure~\ref{fig:d1234}, the intervals generated by applying $D_3$ in our example are $(a_3,0,1)$ from the rule $\dfbeforerule{a_3}{a_2}{a_2}$, $(b_3,3,4)$ from $\dfbeforerule{b_3}{b_2}{b_2}$, and $(a_3,1,3),$ $(a_3,0,3),$ $(a_3,1,4),$ $(a_3,0,4)$ from $\dfbeforerule{a_3}{a_2}{b_2}$.

Now, we introduce rules that filter the intervals produced by the rules in $D_3$ so that only the \emph{minimal} intervals remain.
A minimal interval is one where no other interval (with the same label) is subsumed by it.
The resulting intervals form a contiguous sequence covering all unique timestamps in $\tau$ where their meeting points are the atomic intervals produced by $D_2$.
\begin{equation*}
D_4 = \{ \dfcontainrule{a_4}{a_3}{a_3} \mid a \in \alphabet \}
\end{equation*}

The reason for generating this contiguous sequence of intervals is that we need to transform the input into elements that are \emph{sequentially composable} using data-free \nfer rules.
To understand why, recall our example event trace: $\tau_{\ex} = (a,0)(a,1)(a,2)(b,2)$ $(b,3)(b,4)$.
As seen in Figure~\ref{fig:d1234}, the atomic intervals that result from applying $D_1$ and $D_2$ to this trace are $\{(a_2,0,0),(a_2,1,1),$ $(b_2,3,3),(b_2,4,4)\}$.
Because these intervals do not overlap (they are atomic and have unique timestamps) we can see from Table~\ref{tab:inc-constraints} that the only clock predicate that can match two subsequent intervals (i.e., no labeled interval exists between the end of the first and beginning of the second) is $\beforekw{}$.
The rules in $D_3$, then, do that (match intervals using $\beforekw{}$ rules) but these match both subsequent and non-subsequent intervals.
Applying the rule $\dfbeforerule{a_3}{a_2}{b_2}$, for example, produces $(a_3,1,3),$ $(a_3,0,3),$ $(a_3,1,4),$ and $(a_3,0,4)$.
To match only subsequent intervals requires the rules from step four (in $D_4$).
Applying $\dfcontainrule{a_4}{a_3}{a_3}$ only generates $(a_4,0,1)$ and $(a_4,1,3)$ because they do not \emph{contain} another $a_3$-labeled interval, while $(a_3,0,3)$, $(a_3,1,4)$, and $(a_3,0,4)$ do contain $(a_3,0,1)$ and $(a_3,1,3)$.

At this point, we must discuss what we call the \emph{Ancillary Event Phenomenon}.
Because we must generate sequentially composable intervals to simulate \iac{cfg}, and because these intervals must label the time \emph{between} events, inevitably one event per trace must be unrepresented by such intervals.
Since we choose to label the intervals by their starting event, the final event in the trace with a unique timestamp does not label an interval.
We call this the \emph{ancillary event} in a trace.
In $\tau_{\ex}$, the ancillary event is $(b,4)$.

After applying $D_1 \cup D_2 \cup D_3 \cup D_4$, we have the intervals $(a_4,0,1),(a_4,1,3),$ and $(b_4,3,4)$.
These intervals are now \emph{sequentially composable} because they (uniquely) $\meetkw{}$ at timestamps $1$ and $3$, meaning we can use the $\meetkw{}$ clock predicate to match only the contiguous intervals and no others.

\begin{lemma}\hfill
\begin{enumerate}
    \item Given a trace~$\tau = (\eta_0, t_0)(\eta_1, t_1)\cdots (\eta_{n-1},t_{n-1})$ be a trace over some $\alphabet$, let $\unique(\tau) =  \set{t_{i_0}, t_{i_1}, \ldots, t_{i_{k-1}}}$ with $t_{i_j} < t_{i_{j'}}$ for all $j < j'$, and assume we have $k > 1$.
Then, $\traceinterp{\bigcup_{i=1}^4 D_i}{(\tau)}$ contains all intervals of the form~$((\eta_{{i_{j}}})_4, t_{i_j}, t_{i_{j+1}})$ for $0 \le j < k-1$ and no other interval labeled by identifiers of the form~$a_4$ for some $a \in \Sigma$.

\item Let $w \in \Sigma^*$ and let $\tau_w = \wordtotrace(w)$. Then,\linebreak $\traceinterp{\bigcup_{i=1}^4 D_i}{(\tau_w)}$ contains all intervals of the form $((w_j)_4, j, j+1)$ for $0 \le j < |w|-1$ and no other interval labeled by identifiers of the form~$a_4$ for some $a \in \Sigma$.
\end{enumerate}
\end{lemma}

With the sequentially composable intervals produced by the rules in $D_4$, we now can simulate the productions of the two \acp{cfg}.
Recall that $P$ and $P'$ are the disjoint sets of these productions.
\begin{align*}
\begin{split}
D_5 =\ & \{\ \dfcoinciderule{A}{a_4}{a_4} \mid (A \rightarrow a) \in P \cup P'\ \}\ \cup \\
       & \{\ \dfmeetrule{A}{B}{C} \mid (A \rightarrow BC) \in P \cup P'\ \}
\end{split}
\end{align*}
Unlike the rules from $D_1 \cup D_2 \cup D_3 \cup D_4$, the rules in $D_5$ may contain cycles and must be iterated over until a fixed point is reached.
%

Figure~\ref{fig:d5} shows the result of applying $D_5$ to the running example.
Each rule in $D_5$, shown on the right side of the figure, maps to a production in $P_{\ex}$ and the intervals they produce simulate a derivation for $\tau_{\ex}$ in $G_{\ex}$.
Applying $\dfcoinciderule{A}{a_4}{a_4}$ produces $(A,0,1)$ and $(A,1,3)$, while $\dfcoinciderule{B}{b_4}{b_4}$ produces $(B,3,4)$.
Then, applying $\dfmeetrule{M}{A}{B}$ produces $(M,1,4)$ and applying $\dfmeetrule{S_{\ex}}{A}{M}$ produces $(S_{\ex},0,4)$.
As $S_{\ex}$ is the initial non-terminal for $G_{\ex}$, an $S_{\ex}$-labeled interval in the fixed point indicates that the trace $\tau_{\ex}$ during that interval is in the language of $G_{\ex}$.

\begin{figure*}[t]
    \centering
    \includegraphics[width=.7\textwidth]{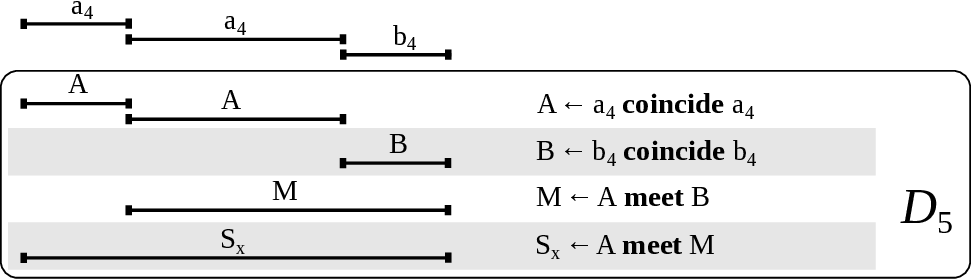}
    \caption{Example of applying step 5 from the proof of Theorem~\ref{thm:undecidable}}
    \label{fig:d5}
\end{figure*}

Next, we show that the data-free \nfer simulation has the desired properties.
We begin by showing correctness for a single grammar.

\begin{lemma}
\label{lemma_correctness_grammar}
Let $G = (\nonterms, \alphabet, P, S)$ be \iac{cfg}.
\begin{enumerate}
    \item \label{lemma_correctness_grammar_trace2word}
Given a trace~$\tau\!=\! (\eta_0, t_0)(\eta_1, t_1)\cdots (\eta_{n-1},t_{n-1})$, let $\unique(\tau)\! =\! \set{t_{i_0}, t_{i_1}, \ldots, t_{i_{k-1}}}$ with $t_{i_j} < t_{i_{j'}}$ for all $j < j'$, and assume we have $k > 1$. Let $0 \le j < j' < k-1$  and $w = \tracetoword(\tau)$, which has length~$k$.  Then,
$(S,t_{i_j},t_{i_{j'+1}}) \in \traceinterp{\bigcup_{i=1}^{5} D_i}{(\tau)}$
iff 
$w_{[j,j']} \in \languageof(G)$.

    \item \label{lemma_correctness_grammar_word2trace}
Given a word $\utword \in \alphabet^*$, fix an identifier $a \in \alphabet$ for the ancillary event, and let $0 \le j < j' < |\utword| $.  Then,
$\utword_{[j,j']} \in \languageof(G)$ 
iff 
$(S,j,j'+1) \in \traceinterp{\bigcup_{i=1}^{5} D_i}{(\ \wordtotrace(\utword \concat a)\ )}$.

\end{enumerate}
\end{lemma}

\begin{proof}
\ref{lemma_correctness_grammar_trace2word}.)
By induction over $j'-j$, showing that for every non-terminal $A \in \nonterms$,
$(A,t_{i_j},t_{i_{j'+1}}) \in \traceinterp{\bigcup_{i=1}^{5} D_i}{(\tau)}$
iff 
$A \xRightarrow{*} w_{[j,j']}$.

\ref{lemma_correctness_grammar_word2trace}.)
By induction over $j'-j$, showing that for every non-terminal $A \in \nonterms$,
$A \xRightarrow{*} \utword_{[j,j']}
$ iff $
(A,j,j'+1) \in \traceinterp{\bigcup_{i=1}^{5} D_i}{(\wordtotrace(\utword \concat a))}$.
\qed\end{proof}

Finally, we check that a word is accepted by both grammars by labeling as $\target$ where the timestamps of any $S$-and-$S'$-labeled intervals are the same.
If any word has a derivation in both $G$ and $G'$, then applying $\bigcup_{i=1}^{6} D_i$ to the corresponding trace will result in an $\target$-labeled interval in the fixed point.
\begin{equation*}
D_6 = \{ \dfcoinciderule{\target}{S}{S'} \}
\end{equation*}
For example, suppose a second grammar $G'_{\ex}$ was introduced accepting the language $a^{+}b^{+}$, where its initial non-terminal was $S'_{\ex}$.
The word $aab$ is in the language of $G'_{\ex}$ and so applying $\bigcup_{i=1}^{5} D_i$ for $G'_{\ex}$ to the trace $\tau_{\ex}$ would yield a fixed point containing the interval $(S'_{\ex},0,4)$.
Since $(S_{\ex},0,4)$ coincides with this interval, applying the rule in $D_6$ will yield $(\target,0,4)$.





This concludes the construction and we can prove of main result of this section.

\begin{proof}[of Theorem~\ref{thm:undecidable}]
Let $w \in \languageof(G) \cap \languageof(G')$ for some nonempty~$w \in \alphabet$.
Then, applying Lemma~\ref{lemma_correctness_grammar}.\ref{lemma_correctness_grammar_word2trace} to both grammars yields \\ 
$(S, 0, |w|),(S', 0, |w|) \in \traceinterp{\bigcup_{i=1}^{6} D_i}{\wordtotrace(w\concat a)}$, and thus $(\target, 0, |w|) \in \traceinterp{\bigcup_{i=1}^{6} D_i}{\wordtotrace(w\concat a)}$, i.e., there is a trace such that $\target$ is generated by applying the rules in $\bigcup_{i=1}^{6} D_i$ to it.

Conversely, let $\tau$ be a trace such that $(\target, j, j') \in \traceinterp{\bigcup_{i=1}^{6} D_i}{\tau}$. 
Thus, we must also have\\
$(S, j, j'),(S', j, j') \in \traceinterp{\bigcup_{i=1}^{6} D_i}{\tau}$. 
Applying Lemma~\ref{lemma_correctness_grammar}.\ref{lemma_correctness_grammar_trace2word} yields an infix of $\tracetoword(\tau)$ that is both in $\languageof(G)$ and $\languageof(G')$, i.e., the intersection of the languages of $G$ and $G'$ is nonempty.

Altogether, we have reduced the undecidable non-emptiness problem for the intersection of two \acp{cfg} to the data-free \nfer satisfiability problem.
\qed\end{proof}

As satisfiability of data-free \nfer is undecidable, we now turn our attention to examining fragments with decidable satisfiability.
We identify two such fragments: Inclusive \nfer, where only inclusive rules are permitted, and Cycle-free \nfer, where specifications can be evaluated without a fixed-point computation.

\subsection{Inclusive Data-Free \texorpdfstring{\nfer}{Nfer} Satisfiability is in \texorpdfstring{\ptime}{PTIME}}
\label{sec:inclusive}


We begin our study with the case where an \nfer specification may contain cycles but only contains inclusive rules.

\begin{theorem}
\label{thm:inclusive}
The data-free, inclusive \nfer satisfiability problem is in \ptime.
\end{theorem}

We show that there is a polynomial-time algorithm that determines if an input trace $\tau$ exists such that an $\target$-labeled interval is in $\traceinterp{D}{\tau}$ for a given specification $D$.
To do this, we show how the satisfiability of a data-free Inclusive-\nfer specification can be proven through an analysis of the rules without guessing a witnessing trace.
This is due to the monotone nature of inclusive \nfer rules: new events added to an input trace only add intervals and cannot invalidate existing ones.
We leverage this fact to show how only two factors influence the satisfiability of cycle-free, inclusive \nfer specifications: producibility from events in $\inputids$ and the requirement of positive duration for some intervals.

To begin, observe that inclusive \nfer rules are monotone in nature.  The interpretation functions $R$, $S$, and $T$ only add intervals; they never remove them.
Furthermore, if the rule is inclusive, $R$ only tests for the existence of intervals; it only tests for non-existence in the case of exclusive rules.
This means that we may always introduce new events into an input trace without needing to keep track of prior results.
The consequence is that ensuring that an $\target$-labeled interval appears in a fixed-point of $\traceinterp{D}{\tau}$ only requires showing that a rule $\delta_T$ exists in $D$ with $\target$ on its left-hand side and that $\delta_T$ may be matched by intervals resulting from $\inputids$-labeled events.
This concept is very similar to graph reachability and we define it here inductively.


\begin{definition}
Let $\inputids$ be a set of input identifiers and $D$ an inclusive \nfer specification. An identifier $\eta$ is \emph{producible} by $D$ iff $\eta \in \inputids$ or if there exists a rule $(\dfnferrule{\eta}{\eta_1}{\oplus}{\eta_2}) \in D$ and both $\eta_1$ and $\eta_2$ are producible by $D$.
\end{definition}


We now prove that satisfiability for an \nfer specification using only rules with the $\beforekw{}$ operator is equivalent to producibility.
We discuss specifications with only $\beforekw{}$-rules here because they allow us to ignore the interaction between events, which have zero duration, and \nfer operators which require positive duration.
We address this complication after proving Lemma~\ref{lemma:before}.

\begin{lemma}
\label{lemma:before}
Given a set of input identifiers $\inputids$ and a target identifier $\target$, an \nfer specification $D_b$ containing only $\beforekw{}$ rules is satisfiable iff $\target$ is producible by $D_b$.
\end{lemma}

\begin{proof}
If $\target \in \inputids$, then $D_b$ is satisfied by the trace~$(\target, 0)$.
If $\target \not\in \inputids$, then for $D_b$ to be satisfiable there must a rule $\delta_T$ in $D_b$ with $\target$ on its left-hand side.
Next, observe from the definition of $\beforekw{}$ in Table~\ref{tab:inc-constraints} that the only requirement of $\delta_T$ to produce an $\target$-labeled interval is that there exist intervals $i_1$ and $i_2$ such that $\Endof{i_1} < \Startof{i_2}$.
Clearly, if $\Idof{i_1} \in \inputids$ and $\Idof{i_2} \in \inputids$ we can create an input trace that satisfies this.
If either identifier on the right-hand side of $\delta_T$ is not in $\inputids$, then apply the same logic inductively for that identifier.

The reverse follows by a similarly straightforward induction: if an $\target$-labeled interval is producible then $D_b$ is satisfiable.
\qed\end{proof}

To permit inclusive operators beyond \beforekw{}, we must address the requirement of \emph{positive duration}.
To see why we need to address positive duration, take, for example, the $\overlapkw{}$ operator.
Again from Table~\ref{tab:inc-constraints}, we see that $\overlapkw{}$ requires that $\Startof{i_1} < \Endof{i_2}$ and $ \Startof{i_2} < \Endof{i_1}$.
If we assume a zero duration for $i_2$, we still must have positive duration for $i_1$: $\Startof{i_1} < \Startof{i_2} = \Endof{i_2} < \Endof{i_1}$, and the same holds for $i_2$ if we assume zero duration for $i_1$.
This means that, for an $\overlapkw{}$-rule to match, at least one interval it matches must have positive duration.
As such, $\overlapkw{}$-rules cannot match two initial intervals (events), as they have zero duration.

Thus, producibility is insufficient to show satisfiability for inclusive-\nfer specifications.
We must augment our definition of what is producible to account for what intervals may be produced with positive duration.

Table~\ref{tab:operator-duration} defines two functions from rules to sets of subsets of identifiers, $\reqmatch : \rulety \rightarrow \powerset{\powerset{\Ident}}$ and $\reqadd : \rulety \rightarrow \powerset{\powerset{\Ident}}$.
The $\reqmatch$ function returns the identifiers that must appear in intervals with positive duration for a given rule to match (produce an interval).
The $\reqadd$ function returns the identifiers that must appear in intervals with positive duration for a given rule to produce an interval with a positive duration.
Both functions return values in Conjunctive-Normal Form (CNF), meaning that at least one element of each set must have positive duration.
For example, $\reqmatch(\dfnferrule{\eta}{\eta_1}{\meetkw{}}{\eta_2}) = \varnothing$ because $\meetkw{}$ can match two intervals with zero duration, but $\reqadd(\dfnferrule{\eta}{\eta_1}{\meetkw{}}{\eta_2}) = \set{\eta_1, \eta_2}$ because at least one of the two intervals it matches must have positive duration for the result to have positive duration.

\begin{table}[ht]
    \caption{Positive duration requirements on $\delta = (\dfnferrule{\eta}{\eta_1}{\oplus}{\eta_2})$ in CNF (for the sake of readability, we omit outer $\{\}$)}
    \label{tab:operator-duration}
    \centering
    \begin{tabular}{c|c|c}
      $\oplus$      & $\reqmatch(\dfnferrule{\eta}{\eta_1}{\oplus}{\eta_2})$ & $\reqadd(\dfnferrule{\eta}{\eta_1}{\oplus}{\eta_2})$\\
      \hline
      \beforekw{}   & $\varnothing$          & $\varnothing$\\
      \meetkw{}     & $\varnothing$          & $\set{\eta_1, \eta_2}$ \\
      \duringkw{}   & $\varnothing$          & $\set{\eta_2}$\\
      \coincidekw{} & $\varnothing$          & $\set{\eta_1},\set{\eta_2}$\\
      \startkw{}    & $\varnothing$         & $\set{\eta_1, \eta_2}$\\
      \finishkw{}   & $\varnothing$         & $\set{\eta_1, \eta_2}$\\
      \overlapkw{}  & $\set{\eta_1, \eta_2}$ & $\set{\eta_1, \eta_2}$\\
      \slicekw{}    & $\set{\eta_1, \eta_2}$ & $\set{\eta_1, \eta_2}$\\
    \end{tabular}
\end{table}


\begin{definition}
Let $\inputids$ be a set of input identifiers  and $D$ an inclusive \nfer specification. An identifier $\eta$ is \emph{positive-duration capable} in $D$ iff there exists a rule $(\dfnferrule{\eta}{\eta_1}{\oplus}{\eta_2}) = \delta \in D$ such that for all $A \in \reqadd(\delta)$, there exists an identifier in $A$ that is positive-duration capable in $D$.
\end{definition} 
Note that $\beforekw{}$-rules~$\delta$ are always positive-duration capable, as we have $\reqadd(\delta) = \emptyset$ for such $\delta$. 
We now define duration-sensitive producibility for identifiers in an inclusive \nfer specification using $\reqmatch$ and the definition of positive-duration capable identifiers.

\begin{definition}
\label{def:duration-sensitive}
Let $\inputids$ be a set of input identifiers  and $D$ an inclusive \nfer specification. An identifier $\eta$ is \emph{duration-sensitive producible} by $D$ iff $\eta \in \inputids$ or if there exists a rule $(\dfnferrule{\eta}{\eta_1}{\oplus}{\eta_2}) = \delta \in D$ such that $\eta_1$ and $\eta_2$ are both duration-sensitive producible by $D$ and for all $M \in \reqmatch(\delta)$, there exists an identifier in $M$ that is positive-duration capable in $D$.
\end{definition}

It should be clear that the consequence of Definition~\ref{def:duration-sensitive} is the following lemma.

\begin{lemma}
\label{lem:duration}
Given a target identifier $\target$, an inclusive \nfer specification $D$ is satisfiable iff $\target$ is duration-sensitive producible by $D$.
\end{lemma}
\begin{proof}
    The proof follows directly from the definitions of satisfiability and duration-sensitive producibility.
    For a specification~$D$ to be satisfiable, given a target identifier~$\target$ and a set of input identifiers~$\alphabet$, there must be an input trace $\tau$ of events with identifiers from $\alphabet$ such that $\traceinterp{D}{\tau}$ contains an interval with identifier $\target$.
    Clearly, if $\target$ is not in $\alphabet$, then there must be a rule $(\dfnferrule{\target}{\eta_1}{\oplus}{\eta_2}) = \delta \in D$ and $D$ must also be satisfiable for target identifiers $\eta_1$ and $\eta_2$.
    For a rule to produce an interval it must meet precisely the conditions in Definition~\ref{def:duration-sensitive}.
    The remainder of the proof is a straightforward induction over the rules in $D$.

    For the other direction, for a target identifier $\target$ to be  duration-sensitive producible by a specification $D$, either $\target$ is in the input identifiers $\alphabet$ or there exists an input trace $\tau$ with events labeled by identifiers in $\alphabet$ such that a $\target$-labeled interval is in $\traceinterp{D}{\tau}$, which is the definition of satisfiability.
    \qed
\end{proof}
Now, we can prove Theorem~\ref{thm:inclusive}.

\begin{proof}
By Lemma~\ref{lem:duration}, satisfiability for data-free inclusive \nfer is equivalent to duration-sensitive producibility.

Algorithm~\ref{alg:incsat} implements Definition~\ref{def:duration-sensitive} to determine the satisfiability of an inclusive \nfer specification.
It works by computing sets of positive-duration capable identifiers ($\durationset$) and duration-sensitive producible identifiers ($\eventset$).
Recall that the topological sort of the strongly-connected components of the graph formed by the rules of an \nfer specification gives a dependency ordering for the evaluation of those rules.
This order is computed on Line~\ref{alg:incsat:topsort} and the strongly-connected components $\mathcal{D}$ are iterated over.
This set is then looped over $|\mathcal{D}|$ times on Lines~\ref{alg:incsat:repeat}~and~\ref{alg:incsat:loop}.
In the worst case, we must loop $|\mathcal{D}|$ times to reach a fixed point.
For each rule, then, Line~\ref{alg:incsat:producible} tests that both the identifiers on the right-hand side are duration-sensitive producible for that rule.
If they are duration-sensitive producible, it then tests on Lines~\ref{alg:incsat:match}~and~\ref{alg:incsat:add} that the rule's duration requirements are met using $\reqmatch$ and $\reqadd$.
If, after iterating over the rules, $\target$ is labeled as duration-sensitive producible, then the $D$ is satisfiable.


\begin{algorithm}
\caption{Algorithm checking Data-Free Inclusive-\nfer satisfiability}
\label{alg:incsat}
\begin{algorithmic}[1]
\REQUIRE{Specification~$D$, identifiers~$\inputids$, target identifier~$\target$}
\STATE{$\durationset \gets \varnothing,\;\; \eventset \gets \inputids$}
\FOR{$\mathcal{D} \in \textit{topsort}(\text{SCC}(D))$} \label{alg:incsat:topsort}
  \FOR{$1 \cdots |\mathcal{D}|$} \label{alg:incsat:repeat}
    \FOR{$\dfnferrule{\eta}{\eta_1}{\oplus}{\eta_2} \in \mathcal{D}$}\label{alg:incsat:loop}
    \IF{$\set{\eta_1, \eta_2} \cap \eventset = \set{\eta_1, \eta_2}$}\label{alg:incsat:producible}
      \IF{$\forall M \in \reqmatch(\dfnferrule{\eta}{\eta_1}{\oplus}{\eta_2}) \where M \cap \durationset \neq \varnothing$}\label{alg:incsat:match}
        \STATE{$\eventset \gets \eventset \cup \set{\eta}$}
      \ENDIF
      \IF{$\forall A \in \reqadd(\dfnferrule{\eta}{\eta_1}{\oplus}{\eta_2}) \where A \cap \durationset \neq \varnothing$}\label{alg:incsat:add}
        \STATE{$\durationset \gets \durationset \cup \set{\eta}$}
      \ENDIF
    \ENDIF
    \ENDFOR
    \STATE{\textbf{end for} \textit{looping over rules in $\mathcal{D}$}}
  \ENDFOR
  \STATE{\textbf{end for} \textit{repeating $|\mathcal{D}|$ times}}
\ENDFOR
\STATE{\textbf{end for} \textit{looping over strongly-connected components of $D$}}
\STATE \textbf{if} $\target \in \eventset$ \textbf{then return} SAT
\STATE \textbf{else} \textbf{return} UNSAT
\end{algorithmic}
\end{algorithm}

We conclude by showing that Algorithm~\ref{alg:incsat} runs in $\bigo{n^2}$ where $n$ is the number of rules:
The topological sort of the strongly-connected components of the graph of $D$ can be computed in linear time.
Then each rule is visited at most $|D|$ times, if all rules are in the same strongly-connected component.
All other operations are sub-linear.
\qed\end{proof}

We illustrate Algorithm~\ref{alg:incsat} in the following example.

\begin{example}
Suppose a target identifier $\target$ and a set of input identifiers $\inputids = \set{a,b}$.
We now consider the result of applying Algorithm~\ref{alg:incsat} to a satisfiable list of inclusive rules with no cycles.
\newcommand{\rulelistsep}{. }
\begin{enumerate}
    \item $\dfnferrule{A}{b}{\beforekw{}}{X}$\rulelistsep
    $X$ is not in $\eventset$ ($\durationset = \varnothing$, $\eventset = \set{a,b}$)
    \item $\dfnferrule{A}{a}{\beforekw{}}{b}$\rulelistsep
    $a$ and $b$ are both in $\eventset$ and $\beforekw{}$ matches and adds duration without any other requirements ($\durationset = \set{A}$, $\eventset = \set{a,b,A}$)
    \item $\dfnferrule{B}{a}{\meetkw{}}{b}$\rulelistsep
    $a$ and $b$ are both in $\eventset$ and $\reqmatch$ for $\meetkw{}$ does not require positive duration, but $\reqadd$ does ($\durationset = \set{A}$, $\eventset = \set{a,b,A,B}$)
    \item $\dfnferrule{\target}{a}{\overlapkw{}}{B}$\rulelistsep
    $a$ and $B$ are both in $\eventset$, but $\reqmatch$ and $\reqadd$ require one of $a,B$ to be positive-duration capable ($\durationset = \set{A}$, $\eventset = \set{a,b,A,B}$)
    \item $\dfnferrule{\target}{A}{\overlapkw{}}{B}$\rulelistsep
    $A$ and $B$ are both in $\eventset$, and $A$ is positive-duration capable, meeting the requirement ($\durationset = \set{A,\target}$, $\eventset = \set{a,b,A,B,\target}$)
\end{enumerate}
\end{example}





\subsection{Cycle-Free Data-Free \texorpdfstring{\nfer}{Nfer} Satisfiability is Decidable}
\label{sec:cyclefree}

\newcommand{\op}{\textbf{op}}

Next, we consider data-free \nfer with inclusive and exclusive rules, but without cycles.
We conjecture that data-free cycle-free \nfer satisfiability is \nexptime-complete but only prove that it is \nexptime-hard and decidable.
Here, we begin with proving decidability obtained by a transformation to \emph{monadic first order logic} (MFO) over strings over the alphabet of identifiers.
MFO allows the quantification of positions in a string and uses the predicates~\myquot{$<$} (the order between positions of the string) and a unary predicate~$\eta(\cdot)$ that holds at each position labeled by $\eta$. We refer to, e.g., \cite{gtw02} for a full definition.

\begin{theorem}
\label{MFO}
    The cycle-free data-free \nfer satisfiability problem is decidable.
\end{theorem}

\begin{proof}
\newcommand{\mfoform}[1]{{\varphi}_{D,{#1}}}
Given a cycle-free specification $D$, we will, for each identifier $\eta$, construct an MFO formula $\mfoform{\eta}(t_0,t_1)$ with free variables $t_0, t_1$, such that $\eta$ is satisfiable with respect to $D$ iff $\mfoform{\eta}(t_0,t_1)$ is satisfiable, i.e. there exists a string (word) $w$ over $\powerset{\Sigma}$ and an assignment $v:\{t_0, t_1\}\rightarrow \{0,\ldots, |w|-1\}$ such that $w,v\models \mfoform{\eta}(t_0,t_1)$. 
Note, that for an event identifier $\eta\in\Sigma$, $\eta(\cdot)$ is a monadic predicate, where $\eta(t)$ evaluates to true in a string $w$ over $2^\Sigma$ at position $t$ if and only if the set of $w$ at position $t$ contains $\eta$.  That is $w,v\models \eta(t)$ if and only if $\eta\in w_{v(t)}$.

Now from $w, v$, where $w,v\models \mfoform{\eta}(t_0,t_1)$, a satisfying input trace for $\eta$ may  be obtained as the concatenation $\tau_{w,v}=( w_{t_0},t_0)( w_{t_0+1},t_0+1)\cdots ( w_{t_1},t_1)$, where for a set 
$w=\{\eta_1,\ldots,\eta_k\}\subseteq\Sigma$ and $t\in \clockty$, we denote by $( w, t)$ the (any) string $(\eta_1,t)(\eta_2,t)\cdots(\eta_k,t)$.

Since $D$ is cycle-free, we may order the identifiers by a topological sort of the directed graph formed by the rules.  The construction of $\mfoform{\eta}(t_0,t_1)$ now proceeds by induction on this order.  In the base case, $\eta$ is an event identifier.  Here $\mfoform{\eta}(t_0,t_1)=
(t_0=t_1\wedge\eta(t_0))$.
For the inductive case, $\mfoform{\eta}(t_0,t_1)$ is obtained by a disjunction of all the rules for $\eta$ in $D$, i.e.:
\[ \mfoform{\eta}(t_0,t_1) = \bigvee\nolimits_{\dfnferrule{\eta}{\eta_1}{\op}{\eta_2}\in D}  \psi_{D,\eta_1\, \op\,\,\eta_2}(t_0,t_1)
\]
 where $\op \in \{\oplus, \unlesskw{}\ \ominus\}$. Here the definition\linebreak $\psi_{D,\eta_1\, \op\,\,\eta_2}(t_0,t_1)$ is obtained using the MFO formulas for $\eta_1$ and $\eta_2$ from the induction hypothesis.
  Here we just give the definition for  two rules, one inclusive
and one exclusive rule, leaving the remaining rules for the reader to provide.
\begin{multline*}
    \psi_{D,\before{\eta_1}{\eta_2}}(t_0,t_1) = \\
    \exists t'_0, t'_1\where t_0\leq t'_1< t'_0\leq t_1 \wedge \mfoform{\eta_1}(t_0,t'_1) \wedge 
           \mfoform{\eta_2}(t'_0,t_1)\\
    \psi_{D,\after{\eta_1}{\eta_2}}(t_0,t_1) = \\
    \mfoform{\eta_1}(t_0,t_1) \wedge 
    \forall t'_0, t'_1\where (t'_0 \leq t'_1 < t_0\leq t_1) 
    \Rightarrow \\
    \neg \mfoform{\eta_2}(t'_0,t'_1)
\end{multline*}

Thus, decidability of MFO satisfiability over finite strings~\cite{buchi,elgot,trakhtenbrot} yields decidability of cycle-free data-free \nfer satisfiability.
\qed\end{proof}

Though the reduction to MFO in Theorem~\ref{MFO}  yields the desired decidability result, it comes with a non-elementary complexity \cite{stockmeyer}. 
In the remainder of this section, we present a \nexptime lower bound. In fact, we conjecture that every satisfiable specification has an exponentially long satisfying trace. This would yield a matching upper bound.
Note that the problem with establishing a tight upper bound comes from the arbitrary nesting of exclusive rules, such that one rule may need to match so as to not produce an interval that blocks another rule.
This possibility blocks constructions that rely on simple induction over the rules and introduces barriers to, for example, game-based solutions.

\begin{theorem}
\label{lemma_cyclefree_lb}
The cycle-free data-free \nfer satisfiability problem is \nexptime-hard.
\end{theorem}

To prove the lower bound, we want to express the existence of an accepting run of a nondeterministic Turing machine with exponential running time as an \nfer satisfiability problem using only cycle-free rules.
Assume the running time on an input~$x$ is bounded by $2^{p(|x|)}$ for some polynomial~$p$, which implies that the length of each configuration is also bounded by $2^{p(|x|)}$.
Thus, by repeatedly halving a configuration, we obtain a tree structure of height~$p(|x|)$. Hence, we can use $\bigo{p(|x|)}$ cycle-free rules to check whether a sequence of $2^{p(|x|)}$ events encodes a valid configuration, aggregating every infix of length~$2^k$, for $k \in \set{0,1,\ldots, p(|x|)}$.

Further, when writing two configurations one after the other, the distance between the encoding of the same cell is at a distance of $2^{p(|x|)}$. 
Hence, we can again use another $\bigo{p(|x|)}$ many cycle-free rules to check whether the second configuration is a successor of the first one.
Finally, as there are only $\bigo{p(|x|)}$ many configurations, we can use final $\bigo{p(|x|)}$ many cycle-free rules to check whether the first configuration is the initial one on $x$, whether for any two successive configurations the second one is a successor of the first one, and whether there is an accepting configuration. 
Thus, checking satisfiability of the \nfer specification is equivalent to checking whether the Turing machine has an accepting run on $x$.

To formalize this, let us fix a nondeterministic Turing machine~$\tm = (Q, \Sigmatm, \Gamma, q_\init, q_\acc, q_\rej, \Delta )$ with exponential running time where $Q$ is the finite set of states, $\Sigmatm$ is the input alphabet, $\Gamma$ is the tape alphabet, $q_\init,q_\acc , q_\rej \in Q$ are the initial, accepting, and rejecting state, respectively, and where $\Delta \subseteq Q \times \Gamma \times Q \times \Gamma \times \set{-1, 0, +1}$ is the transition relation.
As the running time of $\tm$ is exponentially bounded, there is a polynomial~$p\colon \natty \rightarrow \natty$ such that $\tm$ halts after at most $2^{p(|x|)}-2$ steps on any input~$x \in \Sigmatm^*$.
Thus, a configuration of $\tm$ on an input~$x$ can be seen as a word over the alphabet~$\Gamma \cup \Gamma\times Q$ of length~$2^{p(|x|)}-1$ with exactly one letter from $\Gamma\times Q$, whose position encodes the position of the head. 

For technical convenience, we assume $p(|x|) \ge 2$ for all inputs~$x$, which implies that each configuration has at least length~$3$.
Similarly, it will also be convenient to ensure that all accepting runs on an input~$x$ have length~$2^{p(|x|)}$.
Hence, we assume that for every tape alphabet symbol~$a \in \Gamma$, the transition~$(q_\acc,a,q_\acc,a,0)$ is in $\Delta$, but no other transition~$(q_\acc,a,q',b,d)$ for some $q' \in Q$, $b \in \Gamma$ and $d \in \set{-1,0,+1}$, i.e., the accepting state is a sink with self-loops.

Under this assumption, an accepting run of $\tm$ on an input~$x$ can be defined to be a sequence of exactly $2^{p(|x|)}$ configurations~$\conf{1}{}, \conf{2}{}, \ldots , \conf{2^{p(|x|)}}{}$ of length~$2^{p(|x|)}-1$ such that $\conf{1}{}$ is the initial configuration of $\tm$ on $x$, each $\conf{i}{}$ is a successor configuration of $\conf{i-1}{}$, and $\conf{2^{p(|x|)}}{}$ is an accepting configuration.
We will write $\divi \conf{1}{}\divi \conf{1}{}\divi \cdots\divi \conf{2^{p(|x|)}}{}$ for such a run, where $\divi$ is a divider symbol.

Let us also fix an input~$x$ and define $n = p(|x|)$. 
We reduce deciding the existence of an accepting run of $\tm$ on $x$ to the satisfiability problem for cycle-free data-free \nfer, i.e., the specification~$D_{\tm,x}$ we construct is satisfied by an event trace~$\tau$ if and only if a subsequence of $\tau$ encodes an accepting run of $\tm$ on $x$.

We continue by introducing our encoding of configurations, which adds some redundancy to simplify the implementation of the checking of the successor configurations using \nfer rules.
To this end, consider the alphabet~$\Sigma = (\Gamma \cup \Gamma \times Q \cup \set{\divi})^3$. 
In the following, we use roman letters (typically $a$) for elements from $\Sigma$ and Greek letters (typically $\sigma$) for elements from $\Gamma \cup \Gamma \times Q \cup \set{\divi}$.
For $a = (\sigma_1, \sigma_2, \sigma_3) \in \Sigma$, we define $\leftel(a) = \sigma_1$, $\centerel(a) = \sigma_2$, and $\rightel(a) = \sigma_3$.
Now, we say that a word~$a_0 a_1 \cdots a_{m-1} \in \Sigma^*$ is well-formed if we have $\centerel(a_i) = \leftel(a_{i+1})$ and $\rightel(a_i) = \centerel(a_{i+1})$ for all $0 \le i < m-1$.
A well-formed word~$a_0 a_1 \cdots a_{m-1} \in \Sigma^*$ encodes the word~$\centerel(a_0)\centerel(a_1) \cdots \centerel(a_{m-1}) \in (\Gamma \cup \Gamma \times Q \cup \set{\divi})^*$. Let us note that multiple words over $\Sigma$ encode a word over $\Gamma \cup \Gamma \times Q \cup \set{\divi}$, as $\leftel(a_0)$ and $\rightel(a_{m-1})$ are not fixed by our requirements.

Recall that in the proof of Theorem~\ref{thm:undecidable}, where we also encoded words over some alphabet as event traces, we introduced the functions~$\wordtotrace$ and $\tracetoword$ to translate between these two formalisms and discarded events happening at the same time when translating from traces to words.
Here, we adapt these definitions.
Given a word~$\sigma_0 \sigma_1 \cdots \sigma_{m-1} \in (\Gamma \cup \Gamma \times Q \cup \set{\divi})^*$, we define 
\begin{multline*}
\wordtotrace(\sigma_0 \sigma_1 \cdots \sigma_{m-1} ) = 
((\#,\sigma_0,\sigma_1),0)\\
((\sigma_0, \sigma_1, \sigma_2),1)
((\sigma_1, \sigma_2, \sigma_3),2)
\cdots\\
((\sigma_{m-3}, \sigma_{m-2}, \sigma_{m-1}),m-2)
((\sigma_{m-2}, \sigma_{m-1}, \#),m-1)
\end{multline*}
which is well-formed.

Dually, given an event trace~$\tau = (\eta_0, t_0)(\eta_1,t_1)\cdots$ $(\eta_{n-1},t_{n-1})$ over $\Sigma$, recall that we say that $t_i$ is non-unique in $\tau$ if there is a $j \neq i$ with $\eta_i \neq \eta_j$ and $t_j = t_i$, otherwise $t_i$ is unique. 
Also, let $\unique(\tau) = \{ t_{i_0},t_{i_1},\ldots,t_{i_{k-1}} \}$ be the set of unique timestamps in $\tau$ such that $t_{i_j} < t_{i_{j'}}$ for all $j < j'$
Then, we define 
$\tracetoword(\tau) = \eta_{i_{0}} \eta_{i_{1}}\cdots \eta_{i_{k-1}} \in \Sigma^*$ 
and
$\tracetowordcenter(\tau) = \centerel(\eta_{i_{0}}) \centerel(\eta_{i_{1}})\cdots \centerel(\eta_{i_{k-1}}) \in (\Gamma \cup \Gamma \times Q \cup \set{\divi})^*$.

We now begin presenting our \nfer rules that, for the fixed input~$x$ and fixed Turing machine~$\tm$, generate a target identifier (to be introduced later) from a given event trace~$\tau$ over $\Sigma$ if and only if $\tau$ has an infix~$\tau'$ such that $\tracetoword(\tau')$ encodes an accepting run of $\tm$ on $x$.
This involves the following steps corresponding to sets of rules:
\begin{enumerate}
    \item Rules in the sets~$D_1$ to $D_4$ are used to filter out non-unique events as done in the simulation of \acp{cfg} via \nfer rules (see Section~\ref{sec:satisfiability}). This allows to identify event traces and words over $\Sigma$.
    \item Rules in $D_5$ ensure that these words a well-defined (formalized in Proposition~\ref{prop:well-formed}).

    \item Rules in $D_6$ illustrates our main tool to achieve the reduction: a hierarchy of $n$ cycle-free $\meetkw{}$-rules to identify infixes of the form~length~$\blank^{2^n}$ (formalized in Proposition~\ref{prop:infix-hierarchy}).

    \item Rules in $D_7$ check wether an infix contains the accepting state (formalized in Proposition~\ref{prop:infix-contains-accept}). This is later used to check whether an infix encodes an accepting configuration. We present these rules early, as they slightly generalize the rules in $D_6$ and are generalized even further in $D_8$.

    \item Rules in $D_8$ aggregate the first and last letter of an infix of length~$2^n$, as well as whether it contains a state or not (formalized in Proposition~\ref{prop:successor-config}). This is then used to check whether such an infix encodes a configuration and also whether an infix of length~$2\cdot s^n$ encodes two configurations such that the second is a successor configuration of the first.

    \item Rules in $D_9$ use the intervals generated by $D_8$ to check whether an infix of length~$2^n \cdot 2^n$ encodes a run infix (formalized in Proposition~\ref{prop:run-infix}).

    \item Rules in $D_{10}$ and $D_{11}$ to check whether a configuration is the initial one (formalized in Proposition~\ref{prop:initial-config}).

    \item Rules in $D_{12}$ to check whether an infix of length~$2^n \cdot 2^n$ starts with the initial configuration, encodes a run infix, and ends with an accepting configuration (formalized in Proposition~\ref{prop:initial-and-accepts}).

\end{enumerate}

To handle non-unique events in event traces, we reuse the rules presented in the proof of Theorem~\ref{thm:undecidable}.
Specifically, we define
\begin{align*}
    &D_1 = \{ \dfcoinciderule{\spoil}{a}{b} \mid (a, b) \in \Sigma \times \Sigma \wedge a\neq b \},\\
    &D_2 = \{ \dfcontainrule{a_2}{a}{\spoil} \mid a \in \Sigma \},\\
    &D_3 = \{ \dfbeforerule{a_3}{a_2}{b_2} \mid (a,b) \in \Sigma \times \Sigma \}, \text{ and}\\ 
    &D_4 = \{ \dfcontainrule{a_4}{a_3}{a_3} \mid a \in \Sigma \}    
\end{align*}
as before to \myquot{spoil} all events that do not have a unique timestamp and then turn the sequence of events with a unique timestamps into a contiguous sequence of intervals covering all unique timestamps.
We call the application of these four sets of rules the initialization phase.

After the initialization phase, we will no longer be working with the identifiers in $\Sigma$, only with the identifiers in $\set{a_4 \mid a \in \Sigma}$ and fresh ones we introduce below. To not clutter our notation beyond the unavoidable, we drop the subscript~$4$ and (slightly abusively) write $a$ instead of $a_4$ from now on.

As stated above, our goal is to write an \nfer specification that is satisfied exactly by event traces that have infixes that encode accepting runs of $\tm$ on $w$.
We begin by capturing well-formedness: Here, and in the following, we use the special identifier~$\invalid$ to signify that the desired format of the encoding of an accepting run is violated. Later, we will use exclusive rules to test for the non-existence of such errors.
For example, consider the set 
\begin{multline*}
D_5 = \set{ \dfmeetrule{\invalid}{a}{b} \mid (a,b) \in \Sigma\times \Sigma \text{ and }\\ \centerel(a) \neq \leftel(b) \text{ or } \rightel(a) \neq \centerel(b) }
\end{multline*}
of rules that create an $\invalid$-labeled interval if and only if two consecutive non-spoiled identifiers in $\tau$ violate the well-formedness condition. 

In the following, we formalize the effect of the rules in $D_5$. 
Here, and in all following such formalizations, we will, for the sake of readability, only consider event traces without non-unique events. 
We can do so without loss of generality, as we have already seen that the rules in $\bigcup_{i=1}^4 D_i$ filter out non-unique events (see Lemma~\ref{lemma_spoiling}).

\begin{proposition}
\label{prop:well-formed}
\hfill
\begin{enumerate}
    \item Let $w = w_0 \cdots w_{m-1} \in \Sigma^*$, $0 \le j < j' \le m-1$, and let $a \in \Sigma$ be some identifier for the ancillary event. 
    Then, $w_{[j,j']}$ is well-formed if and only if there is no interval of the form~$(\invalid, t, t+2)$ for some $j \le t \le j'-2$ in $\traceinterp{D_5}{(\wordtotrace(w\cdot a))}$.

    \item Let $\tau = (\eta_0, t_0) \cdots (\eta_{m-1}, t_{m-1})$ be an event trace where each event is unique in $\tau$, and let $0 \le j < j' \le m-1$.
    Then, there is no interval of the form $(\invalid, t, t')$ for some $t_{j} \le t\le t' \le t_{j'}$ in $\traceinterp{D_5}{(\tau)}$ if and only if $\tracetoword((\eta_j,t_j)\cdots(\eta_{j'},t_{j'})) $ is well-formed.
\end{enumerate}
\end{proposition}

We now move on to the handling of the transition relation.
As configurations have exponential length, we \myquot{only} need to compare events at an exponential distance to check the correct application of the transition relation. 
This, and similar checks, will be implemented using a hierarchy of $\meetkw{}$-rules creating intervals forming a binary tree of height~$n$, which thus has $2^n$ leaves.
To illustrate the construction, let us begin with the following task:
We want to determine whether an interval spans $2^n$ consecutive non-spoiled events that all have an identifier~$a$ with $\centerel(a) = \blank$. where $\blank$ is the blank symbol of $\tm$.
Later, this will be useful to identify the encoding of the initial configuration. 

To this end, consider the fresh identifiers~$\blankint{k}$ for $1 \le k \le n$ and the set
\begin{align*}D_{6} = &\set{
\dfmeetrule{\blankint{1}}{a}{b} \mid \centerel(a) = \centerel(b) = \blank
    } \cup\\
&   \set{\dfmeetrule{\blankint{k}}{\blankint{k-1}}{\blankint{k-1}} \mid 1 < k \le n}
\end{align*}
of rules. 

\begin{proposition}
\label{prop:infix-hierarchy}
\hfill
\begin{enumerate}
    \item Let $w = w_0 \cdots w_{m-1} \in \Sigma^*$, $0 \le j < j' \le m-1$, and let $a \in \Sigma$ be some identifier for the ancillary event. 
    Then, $\centerel(w_j) \cdots \centerel(w_{j'}) = (\blank)^{2^k}$ if and only if\linebreak $(\blankint{k}, j, j'+1) \in \traceinterp{D_6}{(\wordtotrace(w\cdot a))}$.

    \item Let $\tau = (\eta_0, t_0) \cdots (\eta_{m-1}, t_{m-1})$ be an event trace where each event is unique in $\tau$, and let $0 \le j < j' < m-1$.
    Then, $(\blankint{k}, t_{j}, t_{j'+1}) \in \traceinterp{D_6}{(\tau)}$ if and only if $\tracetowordcenter((\eta_j,t_j)\cdots(\eta_{j'},t_{j'})) = (\blank)^{2^k}$.

\end{enumerate}

\end{proposition}

For other checks, we apply similar constructions using a hierarchy of $\meetkw{}$-rules, but we additionally need to aggregate some information about the events covered by the intervals we are constructing. 
For example, to check whether a configuration is accepting, we need to check whether $2^n$ consecutive events contain at least one event labeled with an identifier~$a$ such that $\centerel(a) \in \Gamma \times \set{q_\acc}$.
To this end, consider the fresh identifiers~$\accid{k}{f}$ for $1 \le k \le n$ and $f \in \set{0,1}$ and the following set~$D_7$ of rules:
\begin{align*}
&\set{ \dfmeetrule{\accid{1}{1}}{a}{b} \mid (a,b) \in \Sigma\times \Sigma \text{ and at least }\\
&\phantom{mmmmmmmmm} \text{ one of $\centerel(a), \centerel(b)$ is in } \Gamma \times \set{q_\acc}  }\  \cup\\
&\set{ \dfmeetrule{\accid{1}{0}}{a}{b} \mid (a,b) \in \Sigma \times \Sigma \text{ and none of }\\
&\phantom{mmmmmmmmmmmmm} \text{$\centerel(a), \centerel(b)$ is in }  \Gamma \times \set{q_\acc} }\  \cup\\
&\set{ \dfmeetrule{\accid{k}{\max\set{f_\ell, f_r}}}{\accid{k-1}{f_\ell}}{\accid{k-1}{f_r}} \mid\\
&\phantom{mmmmmmmmmmm}  1 < k \le n \text{ and } f_\ell, f_r \in \set{0,1}}. 
\end{align*}
Intuitively, the flag~$f$ keeps track of whether the events covered by the interval contain an element from $\Gamma \times \set{q_\acc}$. 

\begin{proposition}
\label{prop:infix-contains-accept}
\hfill
\begin{enumerate}
    \item Let $w = w_0 \cdots w_{m-1} \in \Sigma^*$, $0 \le j < j' \le m-1$, and let $a \in \Sigma$ be some identifier for the ancillary event. 
    Then, $w_{[j,j']}$ has length~$2^k$ and there is a $j \le t \le j'$ with $\centerel(w_{t}) \in \Gamma \times \set{q_\acc}$ if and only if $(\accid{k}{1}, j, j'+1) \in \traceinterp{D_7}{(\wordtotrace(w \cdot a))}$.

    \item Let $\tau = (\eta_0, t_0) \cdots (\eta_{m-1}, t_{m-1})$ be an event trace where each event is unique in $\tau$, and let $0 \le j < j' \le m-1$.
    Then, $(\accid{k}{1}, t_{j}, t_{j'+1}) \in \traceinterp{D_7}{(\tau)}$ if and only if $\tracetowordcenter((\eta_j,t_j)\cdots(\eta_{j'},t_{j'}))$ has length~$2^k$ and there is a $j \le j'' \le j'$ with $\centerel(\eta_{j''}) \in \Gamma \times \set{q_\acc}$.
\end{enumerate}
\end{proposition}

Next, we aggregate, again over all intervals of length $2^k$ for $k \le n$, the identifiers of their first and last event and whether the events covered by the interval contain a $\divi$ and/or whether they contain the encoding of a state. 
Here, we use the set~$S = \set{\varepsilon, Q, Q\divi, \divi Q, Q\divi Q}$ to keep track of that information, as well as some ordering information.
So, consider the fresh identifiers~$\aggid{k}{a}{b}{s}$ for $1 \le k < n$, $a,b \in \Sigma$, and $s \in S$, and the set~$D_8$ containing the following rule
\[
\dfmeetrule{\aggid{1}{a}{b}{s}}{a}{b}
\]
for $(a,b) \in \Sigma\times\Sigma$ and $s \in S$,
if one of the following conditions is satisfied:
\begin{itemize}

    \item $\centerel(a) \in \Gamma$, $\centerel(b) \in \Gamma $, and $s = \varepsilon$.

    \item one of $\centerel(a), \centerel(b) $ is in $\Gamma \times Q$ and the other is in $\Gamma$, and $s = Q$.

    \item one of $\centerel(a), \centerel(b) $ is $\divi$ and the other is in $\Gamma$, and $s = \divi$.

    \item $\centerel(a) \in \Gamma \times Q$, $\centerel(b) = \divi $, and $s = Q\divi$.

    \item $\centerel(a) = \divi $, $\centerel(b) \in \Gamma \times Q$, and $s = \divi Q$. 

    
    
    

    
    
\end{itemize}
As already mentioned above, an identifier of the form $\aggid{1}{a}{b}{s}$ keeps track of the first identifier from $\Sigma$ ($a$), the last identifier from $\Sigma$ ($b$), and whether or not they encode states and or the divider symbol, and in which order these appear. Note that we, for example, do not allow two divider symbols without a state in between or two states without a divider symbol in between, as these patterns can never appear in our encodings of runs.

Also, $D_8$ contains, for each $1 < k < n$, the rule
\[
\dfmeetrule{\aggid{k}{a}{d}{s}}{\aggid{k-1}{a}{b}{s_\ell}}{\aggid{k-1}{c}{d}{s_r}}
\]
for $a,b,c,d \in \Sigma$ and $s_\ell,s_r \in S$ 
if one of the following conditions is satisfied:
\begin{itemize}

    \item $s_\ell = \varepsilon$ and  $s = s_r$.
    
    \item $s_r = \varepsilon$ and  $s = s_\ell$.

    \item $s_\ell = \divi$, $s_r = Q$, and $s = \divi Q$. 

    \item $s_\ell = Q$, $s_r =\divi$, and $s = Q \divi$.

    \item $s_\ell = Q \divi$, $s_r = Q$, and $s = Q \divi Q$.

    \item $s_\ell = Q$, $s_r =\divi Q$, and $s = Q \divi Q$. 
    
\end{itemize}
The rules aggregate the same information described above for longer intervals. Note again that we, for example, do not allow two divider symbols, as there cannot be two such symbols in an infix of length~$2^n$ of the encoding of a run.

Finally, consider the fresh identifier~$\confint$. The set $D_8$ also contains the rule
\[
\dfmeetrule{\confint}{
\aggid{n-1}{a}{b}{s_\ell}
}{
\aggid{n-1}{c}{d}{s_r}
}
\]
if $\centerel(a) = \divi$, $\rightel(d) = \divi$, $s_\ell \cdot s_r =\divi Q$, and the rule
\[
\dfmeetrule{\invalid}{
\aggid{n-1}{a}{b}{s_\ell}
}{
\aggid{n-1}{c}{d}{s_r}
}
\]
if $a$ and $d$ violate the transition relation, e.g.,
we have $a = (\sigma_0, (\sigma_1, q), \sigma_2)$ and $d = (\sigma_0, \sigma_1', (\sigma_2, q'))$ but\linebreak $(q, \sigma_1, q', \sigma_1', +1)$ is not a transition of $\tm$. 
For the sake of readability, we do not spell out all cases in detail here.

\begin{proposition}
\label{prop:successor-config}
\hfill
\begin{enumerate}
    \item Let $w = w_0 \cdots w_{m-1} \in \Sigma^*$, $0 \le j < m- 2\cdot 2^n$, and let $a \in \Sigma$ be some identifier for the ancillary event. 
    Then, $w_{[j,j+2\cdot 2^n]}$ has the form~$\divi c_1 \divi c_2$ for configurations~$c_1,c_2$ where $c_2$ is a successor configuration of $c_1$ if and only if 
    $(\confint, j, j+2^n),(\confint, j+2^n, j+2\cdot2^n) \in \traceinterp{D_5 \cup D_8}{(\wordtotrace(w\cdot a))}$ and 
    $(\invalid, j',j'+2^n) \notin \traceinterp{D_5 \cup D_8}{(\wordtotrace(w\cdot a))}$ for all $j \le j' \le j+2^n$. 
    
    \item Let $\tau = (\eta_0, t_0) \cdots (\eta_{m-1}, t_{m-1})$ be an event trace where each event is unique in $\tau$, and let $0 \le j < m-2\cdot 2^n$.
    Then, $(\confint, t_{j}, t_{j+2^n}),(\confint,t_{j+2^n},t_{j+2\cdot 2^n}) \in \traceinterp{D_5 \cup D_8}{(\tau)}$ and $(\invalid, t_{j'}, t_{j' +2^n}) \notin\traceinterp{D_5 \cup D_8}{(\tau)} $ for all $j \le j' \le j+2^n$  if and only if $\tracetowordcenter((\eta_j,t_j)\cdots(\eta_{j+2\cdot 2^n},t_{j+2\cdot2^n}))$ has the form $\divi c_1 \divi c_2$ for configurations~$c_1,c_2$ where $c_2$ is a successor configuration of $c_1$.

\end{enumerate}
\end{proposition}

Hence, applying the rules defined thus far creates intervals identifying configurations and intervals marking situations where the transition relation is not applied correctly. 
Hence, as a next step, we can piece together configurations into run infixes, as long as the transition relation is always applied correctly.
To this end, consider the fresh identifiers~$\runintprime{1}$ and $\runint{k}$ for $1 \le k \le n$ and the set 
\begin{align*}
D_9 = &\set{
\dfmeetrule{\runintprime{1}}{\confint}{\confint}
}
\ \cup\\
&\set{
\dfcontainrule{\runint{1}}{\runintprime{1}}{\invalid}
}
\ \cup\\
&\set{
\dfmeetrule{\runint{k}}{\runint{k-1}}{\runint{k-1}} \mid 1 < k \le n
}.
\end{align*}
of rules.

\begin{proposition}
\label{prop:run-infix}
\hfill
\begin{enumerate}
    \item Let $w = w_0 \cdots w_{m-1} \in \Sigma^*$, $0 \le j < m- 2\cdot 2^{kn}$, and let $a \in \Sigma$ be some identifier for the ancillary event.
    Then, $w_{[j,j+2\cdot 2^{kn}]}$ has the form~$\divi c_1 \divi \cdots \divi c_{2^k}$ for configurations~$c_1,\ldots, c_{2^k}$ where $c_{j'+1}$ is a successor configuration of $c_{j'}$ for all $1 \le j' < 2^k$ if and only if 
    $(\runint{k}, j, j+2^{kn}) \in \traceinterp{D_5 \cup D_8 \cup D_9}{(\wordtotrace(w\cdot a))}$. 
    
    \item Let $\tau = (\eta_0, t_0) \cdots (\eta_{m-1}, t_{m-1})$ be an event trace where each event is unique in $\tau$, and let $0 \le j < m-2\cdot 2^n$.
    Then, $(\runint{k}, t_{j}, t_{j+2^{kn}}) \in \traceinterp{D_5 \cup D_8 \cup D_9}{(\tau)}$ if and only if $\tracetowordcenter((\eta_j,t_j)\cdots(\eta_{j+ 2^{kn}},t_{j+2^{kn}}))$ has the form~$\divi c_1 \divi \cdots \divi c_{2^k}$ for configurations~$c_1,\ldots, c_{2^k}$ where $c_{j'+1}$ is a successor configuration of $c_{j'}$ for all $1 \le j' < 2^k$.
\end{enumerate}
\end{proposition}

Thus, we can identify runs of length~$2^n$, which is the maximal length we have to consider, as the running time of $\tm$ is bounded exponentially.
Hence, as we have already described how to identify an accepting configuration (see $D_7$), it remains to identify the initial configuration~$(q_\init, x_0) x_1 \cdots x_{|x|-1}\blank^{2^n-(|x|+1)} $, where $\blank$ is the blank symbol of $\tm$. 
To this end, let $n_0$ be the smallest power of two that is larger than $|x|$, which is at most $2\cdot |x|$. Also, let $x'_0 = \divi$, $x_1' = (x_0,q_\init)$, and $x_i' = x_{i-1}$ for $1 < i < |x|$ and $x_i' = \blank$ for $|x| \le i < n_0$, i.e., $x' = x_0' \cdots x_{n_0-1}$ has length~$n_0 = 2^{\log_2(n_0)}$.

Now,  consider the fresh identifiers~$\initint{s}{e}$ with $0 \le s < e \le n_0-1$ and the following set~$D_{10}$ of rules:
\begin{align*}
 &\set{\dfmeetrule{\initint{s}{s+1}}{a}{b} \mid\\
 &\phantom{m}s \in \set{0,1,\ldots, n_0-2}\text{, } \centerel(a) = x_s' \text{, and } \centerel(b) = x_{s+1}'}\ \cup\\
& \set{ \dfmeetrule{\initint{s}{e}}{\initint{s}{m}}{\initint{m+1}{e}} \mid\\
&s,m,e \in \set{0,1,\ldots, n_0-1} \text{ with } s < m \text{ and }  m+1< e }.
\end{align*}

\begin{proposition}
\label{prop:initial-config}
    \hfill
\begin{enumerate}
    \item Let $w = w_0 \cdots w_{m-1} \in \Sigma^*$, $0 \le s < e \le m-1$, and let $a \in \Sigma$ be some identifier for the ancillary event. 
    Then, $\centerel(w_j) \cdots \centerel(w_{j'}) = x_e' \cdots x_s'$ if and only if\linebreak $(\initint{s}{e}, j, j'+1) \in \traceinterp{D_{10}}{(\wordtotrace(w\cdot a))}$.

    \item Let $\tau = (\eta_0, t_0) \cdots (\eta_{m-1}, t_{m-1})$ be an event trace 
    where each event is unique in $\tau$, and let $0 \le j < j' \le m-1$.
    Then, $(\initint{s}{e}, t_{j}, t_{j'}) \in \traceinterp{D_{10}}{(\tau)}$ if and only if $\tracetowordcenter((\eta_j,t_j)\cdots(\eta_{j'},t_{j'})) = x_e' \cdots x_s'$.

    \end{enumerate}

\end{proposition}

Thus, we can identify the initial~$n_0 = 2^{\log (n_0)}$ letters of the initial configuration.
All remaining letters of the initial configuration are $\blank$, i.e., we can use the $\blankint{k}$-labeled intervals introduced in $D_6$.
Consider the fresh identifiers~$\initinttwo{k}$ for $\log(n_0) < k \le n $ and the set $D_{11}$ of rules:
\begin{align*}
    &\set{ \dfmeetrule{\initinttwo{\log(n_0)+1}}{\initint{0}{n_0-1}}{\blankint{\log(n_0)}}  }\cup\\
    &\set{ \dfmeetrule{\initinttwo{k}}{\initinttwo{k-1}}{\blankint{k-1}} \mid 1 < k \le n }.
\end{align*}

\begin{proposition}
\label{prop:initial-and-accepts}
\hfill
\begin{enumerate}

    \item Let $w = w_0 \cdots w_{m-1} \in \Sigma^*$, $0 \le j < m- 2\cdot 2^n$, and let $a \in \Sigma$ be some identifier for the ancillary event. 
    Then, $w_{[j,j+2^n]}$ has the form~$\divi c$ where $c$ is the initial configuration on $w$ if and only if 
    $(\initinttwo{n}, j, j+2^n) \in \traceinterp{D_5 \cup D_6 \cup D_{11}}{(\wordtotrace(w\cdot a))}$. 
    
    \item Let $\tau = (\eta_0, t_0) \cdots (\eta_{m-1}, t_{m-1})$ be an event trace where each event is unique in $\tau$, and let $0 \le j < m-2\cdot 2^n$.
    Then, $(\initinttwo{n},t_{j},t_{j+ 2^n}) \in \traceinterp{D_5 \cup D_6 \cup D_{11}}{(\tau)}$ if and only if $\tracetowordcenter((\eta_j,t_j)\cdots(\eta_{j+ 2^n},t_{j+2^n}))$ has the form~$\divi c$ where $c$ is the initial configuration on $w$.

\end{enumerate}

\end{proposition}

Now, all that remains is to put all pieces (intervals) together. 
Generating an interval labeled by the target identifier~$\target$ is possible if and only if there is a run infix of length~$2^n$ that starts with the initial configuration and ends with an accepting one. Let $D_{12}$ be the set containing the following two rules for fresh identifiers~$\target'$, and $\target$.
\begin{align*}
 &\dfstartrule{\target'}{\initinttwo{n}}{\runint{n}}\\
 &\dffinishrule{\target}{\accid{n}{1}}{\target'}\\
\end{align*}

Now, we can complete the prof of our lower bound (Theorem~\ref{lemma_cyclefree_lb})
\begin{proof}
Combining all the facts collected about the rules in $D_{\tm,x} = \bigcup_{i=1}^{12} D_i$, we obtain that $\tm$ accepts $x$ if and only if there is an input trace~$\tau$ such that there is an $\target$-labeled interval in $\traceinterp{D_{\tm,x}}{\tau}$.
As $\tm$ is a nondeterministic Turing machine with exponential running time, and $D_{\tm,x}$ is acyclic, we obtain that satisfiability for cycle-free data-free \nfer is $\nexptime$-hard.
\qed\end{proof}

\section{Evaluation of Data-Free \texorpdfstring{\nfer}{Nfer}}
\label{sec:evaluation}

The evaluation problem for \nfer asks, given a specification~$D$, a trace~$\tau$ of events, and a target identifier~$\target$, is there an $\target$-labeled interval in $\traceinterp{D}{\tau}$?
The problem has been extensively studied in the presence of data, with complexities ranging from undecidable (for arbitrary data and cycles in the rules) to $\ptime$ (for finite data under the minimality constraint). We refer to~\cite{kauffman2024complexity} for an overview of the results.
One case that has not been considered thus far is the complexity of the evaluation problem for data-free \nfer. 

Obviously, the result from~\cite{kauffman2024complexity} for finite-data covers the case of data-free specifications, 
but without the ``minimality'' meta-constraint that artificially limits the size of the result, evaluation with only inclusive rules
is $\pspace$-complete (without cycles) and respectively $\exptime$-complete (with cycles).
Here, we show that these results depend on the availability of (finite) data:  data-free \nfer can be evaluated in polynomial time (even without minimality).

\begin{theorem}
The evaluation problem for data-free \nfer is in \ptime.
\end{theorem}

\begin{proof}
Consider an input consisting of a specification~$D$, a trace~$\tau$ of events, and a target identifier~$\target$, and let $k$ be the number of unique timestamps in $\tau$.

Recall that an interval is completely specified by its identifier in $\Ident$ and its starting and ending timestamp. 
Hence, as the application of rules does not create new timestamps (cf.\ Table~\ref{tab:inc-constraints} and Table~\ref{tab:exc-constraints}), the number of intervals in $\traceinterp{D}{\tau}$ is bounded by $k^2 |\Ident|$.

Checking if an interval is generated by a specification can also be done in polynomial time.
Whether a single inclusive rule is applicable to two intervals can be checked in constant time.
For exclusive rules, checking if an interval should be generated requires testing at most every interval, the number of which is polynomial in the size of the input.
Thus, one can compute $\traceinterp{D}{\tau}$ in polynomial time and then check whether it contains an $\target$-labeled interval.
\qed\end{proof}

\section{Conclusion and Future Work}
\label{sec:conclusion}

We have studied the complexity of the satisfiability and evaluation problems for Data-free \nfer.
We proved that the evaluation problem is in \ptime and the satisfiability problem is undecidable in the general case, but decidable for cycle-free specifications and in \ptime for specifications with only inclusive rules.

It is interesting that, although much of the complexity of \nfer evaluation comes from data, the satisfiability problem is still hard without it.
As we have seen from prior works~\cite{kauffman2022complexity,kauffman2024complexity}, complexity can often be mitigated by restricting the language to meaningful fragments. 
In the case of data-free \nfer, removing exclusive rules is sufficient to obtain \ptime complexity.

There are still open questions around the complexity of \nfer that may be interesting.
We showed that satisfiability for data-free \nfer is decidable and \nexptime-hard for cycle-free specifications, but we do not prove a tight bound and we conjecture it to be \nexptime-complete.

Another open question is if satisfiability is decidable for restricted cases of \nfer \emph{with data}, for example if specifications are cycle-free and data is finite.
It is also interesting to consider if there are restrictions to data that have not yet been considered that might lead to better complexity, for example restricting to a single datum per interval and only equality tests.

We are also interested in the complexity of \emph{monitoring} \nfer.
Here and in other works, \nfer is presented with an offline semantics.  A naïve monitoring algorithm might simply recompute produced intervals each time a new event arrives, but we suspect that better monitoring complexity can be achieved without requiring assumptions beyond temporal ordering.
We hope that this work inspires others to examine the complexity of other modern \ac{rv} languages.

Finally, another interesting direction for future work is to consider the universality problem, i.e., the question of whether every specification allows to derive an interval with a given identifier. 
For practitioners deploying \nfer in a realistic scenario, universality is likely to indicate a mistake in their specification since an identifier that is produced by every input is not likely to assist in trace comprehension.
As such, the universality problem has application as a way to catch bugs before deployment.

\bibliographystyle{plain}
\bibliography{related,sean-cv}

\end{document}

%% file: acronyms.tex
\let\IDeclareAcronym\DeclareAcronym
\renewcommand{\DeclareAcronym}[2]{%
 \IDeclareAcronym{#1}{%
  #2,foreign-plural={} 
  }
}
\DeclareAcronym{atl}{
  short = ATL,
  long = Allen's Temporal Logic,
  short-indefinite = an,
  long-indefinite = an
}
\DeclareAcronym{cep}{
  short = CEP,
  long  = Complex Event Processing
}
\DeclareAcronym{cfg}{
  short = CFG,
  long = Context-Free Grammar
}

\DeclareAcronym{hs}{
  short = HS,
  long  = Halpern and Shoham's logic of intervals,
  short-indefinite = an
}
\DeclareAcronym{hdc}{
  short = HDC,
  long  = Hybrid Duration Calculus,
  short-indefinite = an
}
\DeclareAcronym{qre}{
  short = QRE,
  long = Quantified Regular Expression
}
\DeclareAcronym{rv}{
  short = RV,
  long  = Runtime Verification,
  short-indefinite = an
}
\DeclareAcronym{srv}{
  short = SRV,
  long  = Stream Runtime Verification,
  short-indefinite = an
}
\DeclareAcronym{jpl}{
  short = JPL,
  long = Jet Propulsion Laboratory
}

%% file: macros.tex
\newcommand{\myquot}[1]{``#1''}

\newcommand{\set}[1]{\{#1\}}

\newcommand{\true}{\mathit{true}}
\newcommand{\false}{\mathit{false}}
\newcommand{\bigo}[1]{O(#1)}

\usepackage{url}
\makeatletter
\AtBeginDocument{%
  \def\doi#1{\url{https://doi.org/#1}}}
\makeatother

\usepackage{needspace}
\newcommand{\theorempreskip}{3\baselineskip} 

\renewenvironment{theorem}{%
  \needspace{\theorempreskip}
  \begin{oldtheorem}%
}{%
  \end{oldtheorem}%
}

\renewenvironment{proposition}{%
  \needspace{\theorempreskip}%
  \begin{oldproposition}%
}{%
  \end{oldproposition}%
}

\newcommand{\pspace}{\textsc{PSpace}\xspace}
\newcommand{\ptime}{\textsc{PTime}\xspace}

\newcommand{\exptime}{\textsc{ExpTime}\xspace}
\newcommand{\expspace}{\textsc{ExpSpace}\xspace}
\newcommand{\nexptime}{\textsc{NExpTime}\xspace}

\newcommand{\tm}{\mathcal{M}}

\newcommand{\nfer}{{\tt nfer}\xspace}
\newcommand{\Nfer}{{\tt Nfer}\xspace}

\newcommand{\etal}{et al.~}

\renewcommand{\mid}{:}

\newcommand{\natty}{\mathbb{N}}
\newcommand{\boolty}{\mathbb{B}}

\newcommand{\powerset}[1]{2^{#1}}
\newcommand*{\Cdot}{\raisebox{-0.25ex}{\scalebox{1.6}{.}}}
\newcommand{\where}[1][1]{\Cdot\ifnum #1=0{}\else{\foreach \n in {1,...,#1}{\ }}\fi} 

\newcommand{\concat}{\cdot}

\newcommand{\alphabet}{\Sigma}
\newcommand{\Ident}{\mathcal{I}}

\newcommand{\utword}{w}

\newcommand{\mapty}{\mathbb{M}}
\newcommand{\emptymap}{\{\ \}}
\makeatletter
\newcommand{\pto}{}
\newcommand{\pgets}{}

\DeclareRobustCommand{\pto}{\mathrel{\mathpalette\p@to@gets\to}}
\DeclareRobustCommand{\pgets}{\mathrel{\mathpalette\p@to@gets\gets}}

\newcommand{\p@to@gets}[2]{%
  \ooalign{\hidewidth$\m@th#1\mapstochar\mkern5mu$\hidewidth\cr$\m@th#1\to$\cr}%
}
\makeatother

\newcommand{\Idkw}{\textit{id}}
\newcommand{\Startkw}{\textit{start}}
\newcommand{\Endkw}{\textit{end}}
\newcommand{\Mapkw}{\textit{map}}
\newcommand{\Idof}[1]{\Idkw(#1)}
\newcommand{\Startof}[1]{\Startkw(#1)}
\newcommand{\Endof}[1]{\Endkw(#1)}

\newcommand{\clockty}{\natty}

\newcommand{\spec}{\mathcal{D}}
\newcommand{\target}{\eta_{T}\xspace}
\newcommand{\inputids}{\alphabet}
\newcommand{\nonterms}{V}
\newcommand{\spoil}{\texttt{SPOIL}}
\newcommand{\wordtotrace}{\textit{TRACE}}
\newcommand{\tracetoword}{\textit{WORD}}
\newcommand{\tracetowordcenter}{\textit{WORD}^\centerel}
\newcommand{\unique}{\textit{UNIQ}}

\newcommand{\languageof}{\mathcal{L}}

\newcommand{\beforekw}{{\bf before}}
\newcommand{\duringkw}{{\bf during}}
\newcommand{\slicekw}{{\bf slice}}
\newcommand{\meetkw}{{\bf meet}}
\newcommand{\coincidekw}{{\bf coincide}}
\newcommand{\startkw}{{\bf start}}
\newcommand{\finishkw}{{\bf finish}}
\newcommand{\overlapkw}{{\bf overlap}}

\newcommand{\unlesskw}{{\bf unless}}
\newcommand{\containkw}{{\bf contain}}
\newcommand{\followkw}{{\bf follow}}
\newcommand{\afterkw}{{\bf after}}

\newcommand{\before}[2]{#1 \ \beforekw\  #2}

\newcommand{\after}[2]{#1 \ \unlesskw \ \afterkw\ #2}

\newcommand{\dfnferrule}[4]{#1 \leftarrow #2\:#3\ #4}
\newcommand{\dfcoinciderule}[3]{\dfnferrule{#1}{#2}{\coincidekw{}}{#3}}
\newcommand{\dfcontainrule}[3]{\dfnferrule{#1}{#2}{\unlesskw{} \ \containkw{}}{#3}}
\newcommand{\dfbeforerule}[3]{\dfnferrule{#1}{#2}{\beforekw{}}{#3}}
\newcommand{\dfmeetrule}[3]{\dfnferrule{#1}{#2}{\meetkw{}}{#3}}
\newcommand{\dffollowrule}[3]{\dfnferrule{#1}{#2}{\unlesskw{} \ \followkw{}}{#3}}

\newcommand{\dfstartrule}[3]{\dfnferrule{#1}{#2}{\startkw{}}{#3}}
\newcommand{\dffinishrule}[3]{\dfnferrule{#1}{#2}{\finishkw{}}{#3}}

\newcommand{\durationset}{I_{+}}
\newcommand{\eventset}{I_{\inputids}}
\newcommand{\reqmatch}{\text{match}_{+}}
\newcommand{\reqadd}{\text{add}_{+}}

\makeatletter
\newcommand{\shorteq}{%
  \settowidth{\@tempdima}{n}
  \resizebox{\@tempdima}{\height}{=}%
}
\makeatother

\newcommand{\eventty}{\mathbb{E}}
\newcommand{\tracety}{\eventty^*}
\newcommand{\intervalty}{\mathbb{I}}

\newcommand{\event}{\varepsilon}

\newcommand{\pool}{\pi}
\newcommand{\selection}{\partial}

\newcommand{\minimality}{\textbf{minimality}\,}

\newcommand{\interp}[2]{R[ \! [ #1 ] \! ] \: #2}

\newcommand{\seminterp}[2]{S[ \! [ #1 ] \! ] \: #2}
\newcommand{\traceinterpX}[4][]{%
\if\relax\detokenize{#4}\relax
  \if\relax\detokenize{#3}\relax
    T[ \! [ #2 ] \! ]%
  \else
  T\!{\text{\tiny#3}}[ \! [ #2 ] \! ]%
  \fi
\else
  \if\relax\detokenize{#3}\relax
  T\!{\overset{{\scriptscriptstyle#4}}{\text{\tiny#3}}}[ \! [ #2 ] \! ]%
  \else
  T\!{\overset{{\scriptscriptstyle#4}}{\text{\tiny#3}}}[ \! [ #2 ] \! ]%
  \fi
\fi%
\ifthenelse{\isempty{#1}}{}{\: #1}}

\newcommand{\traceinterp}[3][]{\traceinterpX[#3]{#2}{}{#1}}

\makeatletter
\newcommand{\xRightarrow}[2][]{\ext@arrow 0359\Rightarrowfill@{#1}{#2}}
\makeatother

\newcommand{\poolty}{\mathbb{P}}
\newcommand{\rulety}{\Delta}

\newcommand{\initf}[1]{\textit{init}(#1)}

\lstdefinelanguage{semantics}{
  morekeywords=
    {if,then,else,let,in,least,such,that,unless,it,exists}, 
  otherkeywords={}, 
  literate=
    {[:}{{$[\![$}}2
    {:]}{{$]\!]$}}2
    {Pool}{{$\poolty$}}2
    {Interval}{{$\intervalty$}}2
    {Map}{{$\mapty$}}2
    {Clock}{{$\clockty$}}4
    {Trace}{{$\tracety$}}2
    {Nats}{{$\natty$}}2
    {emptyset}{{$\varnothing$}}2
    {forall}{{$\forall$}}2
    {overlineforall}{{$\overline{\forall}$}}2
    {not}{{$\neg$}}2
    {isin}{{$\in$}}1
    {isnotin}{{$\notin$}}2
    {union}{{$\cup$}}2
    {UNION}{{$\bigcup$}}2
    {intersect}{{$\cap$}}2
    {superset}{{$\supset$}}2
    {subseteq}{{$\subseteq$}}2
    {max}{{$\max$}}2
    {min}{{$\min$}}2
    {Exists}{{$\exists$}}2
    {selection}{{$\selection$}}2
    {minimality}{{\minimality}}8
    {:-}{{$\where$}}2  
    {&&}{{$\wedge$}}2
    {||}{{$\vee$}}2
    {!=}{{$\neq$}}2
    {>=}{{$\geq$}}2
    {=<}{{$\leq$}}2
    {<.}{{$\langle$}}2
    {.>}{{$\rangle$}}2
    {...}{{$\cdots$}}2
    {oplus}{{$\oplus$}}2
    {bottom}{{$\bot$}}2
    {emptymap}{{$\emptymap$}}2
    {ominus}{{$\ominus$}}2
    {::=}{{$\doteq$}}2
    {<-}{{$\leftarrow$}}2
    {->}{{$\rightarrow$}}2
    {<=}{{$\Leftarrow$}}2
    {=>}{{$\Rightarrow$}}2
    {!=>}{{$\not\Rightarrow$}}2
    {|->}{{$\mapsto$}}2
    {Mapof}{{$\Mapkw$}}3
    {Idof}{{$\Idkw$}}2
    {Startof}{{$\Startkw$}}4
    {Endof}{{$\Endkw$}}2
    {mapkw}{{\bf map}}4
    {wherekw}{{\bf where}}6
    {RuleSem}{{$R$}}2
    {RuleSel}{{$R_{\selection}$}}2
    {RuleMin}{{$R_{\scriptscriptstyle{min}}$}}2
    {SpecSem}{{$S$}}2
    {TraceSem}{{$T$}}2
    {Lambda}{$\Lambda$}2
    {Delta}{$\Delta$}2
    {Phi}{$\Phi$}2
    {Psi}{$\Psi$}2
    {Upsilon}{$\Upsilon$}2
    {pi}{$\pi$}2
    {pi1}{$\pi_{1,1}$}2
    {pi2}{$\pi_2$}2
    {piI}{$\pi_i$}2
    {piO}{$\pi_{i+1,1}$}4
    {piJ}{$\pi_{i,j}$}3
    {piK}{$\pi_{i,j+1}$}5
    {piN}{$\pi_n$}2
    {piprime}{$\pi^{\prime}$}2
    {eta}{$\eta$}1
    {etaP}{$\eta^{\prime}$}1
    {eta1}{$\eta_1$}2
    {eta2}{$\eta_2$}2  
    {eta3}{$\eta_3$}2
    {tau}{{$\tau$}}2
    {eps}{{$\event$}}2
    {eps1}{{$\event_1$}}2
    {eps2}{{$\event_2$}}2
    {eff}{{$f$}}2
    {gee}{{$g$}}2
    {s0}{{$s$}}1
    {e0}{{$e$}}1
    {J0}{{$J$}}1
    {K0}{{$K$}}1
    {L0}{{$L$}}1
    {M0}{{$M$}}1
    {MP}{{$M^{\prime}$}}1
    {s1}{{$s_1$}}2
    {s2}{{$s_2$}}2
    {s3}{{$s_3$}}2
    {sP}{{$s^{\prime}$}}1
    {e1}{{$e_1$}}2
    {e2}{{$e_2$}}2
    {e3}{{$e_3$}}2
    {eP}{{$e^{\prime}$}}1
    {i0}{{$i$}}1
    {i1}{{$i_{1}$}}1
    {i2}{{$i_{2}$}}1
    {init}{{$init$}}2
    {topsort}{{\textit{topsort}}}4
    {lambda1}{{$\lambda_1$}}2
    {lambdaN}{{$\lambda_n$}}2
    {LambdaSet}{{$\powerset{\Lambda}$}}2
    {rule}{{$\delta$}}2
    {delta1}{{$\delta_1$}}2
    {delta2}{{$\delta_2$}}2
    {deltaN}{{$\delta_n$}}2
    {Di}{{$D_i$}}2
    {D1}{{$D_1$}}2
    {Dl}{{$D_\ell$}}2
    {,,,}{{$,\ldots,$}}3
    {Comps}{{$\spec$}}2
    {SIZEl}{{$_{\ell+1,1}$}}2
    {ENN}{$n$}1
    {DeltaList}{{$\Delta^{*}$}}2
    {RuleList}{{$\mathbf{d}$}}2
    {AllI}{{$_{i > 0}$}}3
    {AllJ}{{$_{j > 0}$}}3,
  sensitive=true, 
  morecomment=[l]{//}, 
  morecomment=[s]{/*}{*/},
  escapeinside={(*}{*)},
  stringstyle=\ttfamily,
  aboveskip=5mm,
  belowskip=5mm,
  showstringspaces=false,
  columns=flexible,
  morestring=[b]",
  numberstyle=\scriptsize,
  moredelim=[is][\em]{@}{@}
}

\newcommand{\semantics}{\lstset{language=semantics,backgroundcolor=\color{white}, frame=none,basicstyle={\normalsize},aboveskip=.25em,belowskip=.25em}}
\newcommand{\isem}{\semantics\lstinline}

\usetikzlibrary{positioning,arrows,backgrounds,automata,decorations.shapes,decorations.pathmorphing,decorations.pathreplacing,decorations.markings,decorations.text,fit,matrix}
\tikzset{
  initial text=$ $ 
}

\newcommand{\conf}[2]{c_{#1}^{#2}}
\newcommand{\divi}{\#}

\newcommand{\init}{I}
\newcommand{\acc}{A}
\newcommand{\rej}{R}
\newcommand{\invalid}{\texttt{INVLD}}
\newcommand{\Sigmatm}{\Sigma_\tm}

\newcommand{\accid}[2]{\texttt{ACC}_{#1}[#2]}

\newcommand{\aggid}[4]{\texttt{AGG}_{#1}[#2,#3,#4]}

\newcommand{\centerel}{c}
\newcommand{\rightel}{r}
\newcommand{\leftel}{\ell}

\newcommand{\confint}{\texttt{CONF}}
\newcommand{\runintprime}[1]{\texttt{RUN}_{#1}'}
\newcommand{\runint}[1]{\texttt{RUN}_{#1}}
\newcommand{\initint}[2]{\texttt{INPUT}[#1,#2]}
\newcommand{\initinttwo}[1]{\texttt{INIT}_{#1}}
\newcommand{\blankint}[1]{\texttt{BLNK}_{#1}}
\newcommand{\blank}{\text{\textvisiblespace}}